\documentclass[final,onecolumn,12pt]{IEEEtran}
\usepackage{graphicx,psfrag,epsfig,epsf,latexsym,hhline,amsmath,amssymb,multirow}
\usepackage[usenames,dvipsnames]{pstricks}
\usepackage{pst-plot}
\usepackage{pstricks-add}
\usepackage{hyperref}
\usepackage{graphicx}
\usepackage{amsthm}
\usepackage{algorithm}
\usepackage{algpseudocode}

\usepackage{blindtext}
\usepackage{etoolbox}

\input{Jerry.def}

\begin{document}
\title{Multistage Compute-and-Forward with Multilevel Lattice Codes Based on Product Constructions}
\author{Yu-Chih Huang, Krishna R. Narayanan, and Nihat Engin Tunali\\
Department of Electrical and Computer Engineering \\
Texas A\&M University\\
{\tt\small {\{jerry.yc.huang@gmail.com, krn@ece.tamu.edu, engintunali@neo.tamu.edu\}} }}

\maketitle

\begin{abstract}
A novel construction of lattices is proposed. This construction can be thought of as Construction A with codes that can be represented as the Cartesian product of $L$ linear codes over $\mathbb{F}_{p_1},\ldots,\mathbb{F}_{p_L}$, respectively; hence, is referred to as the product construction. The existence of a sequence of such lattices that are good for quantization and Poltyrev-good under multistage decoding is shown. This family of lattices is then used to generate a sequence of nested lattice codes which allows one to achieve the same computation rate of Nazer and Gastpar for compute-and-forward under multistage decoding, which is referred to as lattice-based \textit{multistage compute-and-forward}.

Motivated by the proposed lattice codes, two families of signal constellations are then proposed for the separation-based compute-and-forward framework proposed by Tunali \textit{et al.} together with a multilevel coding/multistage decoding scheme tailored specifically for these constellations. This scheme is termed separation-based \textit{multistage compute-and-forward} and is shown having a complexity of the channel coding dominated by the greatest common divisor of the constellation size (may not be a prime number) instead of the constellation size itself.
\end{abstract}


\section{Introduction}

Compute-and-forward is a novel information forwarding paradigm in wireless communications in which relays in a network decode functions of signals transmitted from multiple transmitters and forward them to a central destination. If these functions are chosen as linear integer combinations, lattice codes are one of the most natural ways to implement a compute-and-forward scheme since a lattice is closed under addition. If the channel state information is not available at the transmitters, compute-and-forward can be implemented effectively by allowing the relay to choose integer coefficients depending on the channel coefficients and signal-to-noise ratio (SNR). Such a scheme which uses lattices over integers has been analyzed by Nazer and Gastpar for AWGN networks in \cite{nazer2011CF} where achievable computation rates were derived. Based on this approach, in \cite{Engin12}, Tunali \textit{et al.} considered the use of lattices over Eisenstein integers for the compute-and-forward paradigm and successfully extended the result on achievable rates in \cite{nazer2011CF} to lattices over Eisenstein integers.

The lattice codes adopted in \cite{nazer2011CF} are based on those generated by Construction A \cite{LeechSloane71} \cite{conway1999sphere} whose decoding complexity typically depend on decoding of the underlying linear codes. One main drawback of the Construction A lattices is that in order to be Poltyrev-good, the underlying linear codes have to be implemented over very large prime fields which in turn result in high decoding complexity. To alleviate this drawback, in the first part of the paper, inspired by Theorem 2 in \cite{Feng10}, we propose a novel lattice construction called product construction that can be thought of as Construction A \cite{LeechSloane71} with codes which can be represented as the Cartesian product of $L$ linear codes over $\mbb{F}_{p_1},\ldots,\mbb{F}_{p_L}$, respectively. This construction is shown to be able to generate sequences of lattices which are Poltyrev-good under multistage decoding and good for mean-squared error (MSE) quantization. We then generate a sequence of nested lattice codes by extending the result by Ordentlich and Erez in \cite{ordentlich_erez_simple} to the proposed lattices. This sequence of nested lattice codes is adopted for the compute-and-forward problem and a novel strategy called multistage compute-and-forward is proposed which can recover the achievable computation rates in \cite{nazer2011CF} using multistage decoding.

After establishing the information-theoretic results, one important next step would be making progress toward the construction of practically implementable coding schemes for the compute-and-forward paradigm. In \cite{Feng10}, Feng, Silva, and Kschischang have extended the framework in \cite{nazer2011CF} towards the design of efficient and practical schemes via an algebraic approach. In \cite{Engin11SC}, a scheme based on the concatenation of signal codes \cite{SC} with low-density parity check (LDPC) codes have been implemented to compute-and-forward. One of the major drawbacks of this scheme is the substantially high decoding complexity resulting from the fact that for such lattices, the shaping and channel coding are coupled together. This hinders optimal decoding as the dimensionality grows and also results in inseparable shaping and coding gains. In \cite{Engin12}, Tunali \textit{et al.} have proposed a framework that allows the separation of channel coding and data modulation. This scheme is motivated by Construction A \cite{LeechSloane71} over $\mbb{F}_p$, which uses a linear code over $\mbb{F}_p$ in conjunction with a constellation which is carefully cropped from the integers (similarly Gaussian integers and Eisenstein integers) with $p$ elements. In contrast to the schemes in the existing literature, this separation-based scheme has enabled one to separately improve the coding gain and shaping gain, thus resulting in increased computation rates. This separation has also allowed one to keep the constellation size small so that optimal demodulation is feasible.

One of the main drawbacks of this scheme is that the decoding complexity increases dramatically with $p$ the constellation size; hence, the computational complexity of this scheme is quite high in the high rate regime. In the second part of the paper, we aim to construct coding schemes with lower decoding complexity while still maintaining desirable properties such as the ability to perform compute-and-forward. Motivated by the successes of using the proposed lattices for lattice-based multistage compute-and-forward, we propose a novel strategy called separation-based multistage compute-and-forward in which we propose two families of signal constellations together with a multilevel coding/multistage decoding scheme specifically tailored for these constellations so that the complexity of the channel coding is dominated by the greatest divisor of the constellation size (may not be a prime number) instead of the constellation size itself. This substantially reduces the decoding complexity for a given size of the constellation (or, equivalently, asymptotic rate) and hence makes the proposed scheme more practically implementable than the existing ones \cite{nazer2011CF} \cite{Engin12} \cite{Feng10}. It should be noted that although we particularly focus on compute-and-forward, the proposed construction of lattices and the proposed scheme are suitable for many other applications that use the lattice structure such as integer-forcing linear receivers \cite{zhan10}, precoded compute-and-forward \cite{Hong13}, lattice interference alignment \cite{ordentlich11} \cite{ordentlich12}, etc.

\subsection{Organization}
The paper is organized as follows. In Section~\ref{sec:problem}, we state the compute-and-forward relay network of Nazer and Gastpar \cite{nazer2011CF} and the problem of maximizing the computation rates. In Section~\ref{sec:prelim}, some background on algebra is provided and both the lattice-based compute-and-forward \cite{nazer2011CF} and the separation-based compute-and-forward scheme \cite{Engin12} are reviewed. We then present the proposed product construction of lattices, show its goodness, and compare it with Construction D in Section~\ref{sec:prod_const}. Lattices based on the proposed product construction are then used to generated nested lattice codes for lattice-based compute-and-forward and similar computation rates as those in \cite{nazer2011CF} are derived in Section~\ref{sec:lattice_CF}. Founded upon the product construction lattices, the proposed constellations and the proposed multilevel coding/multistage decoding for compute-and-forward are given in Section~\ref{sec:const} and Section~\ref{sec:encode_decode}, respectively. The achievable computation rates are computed using Monte-Carlo techniques in Section~\ref{sec:simulation}. Section~\ref{sec:conclude} concludes the paper.

\subsection{Notations}
Throughout the paper, we use $\mbb{Z}$, $\mbb{N}$, $\mbb{R}$, and $\mbb{C}$ to represent the set of integers, natural numbers, real numbers, and complex numbers, respectively. We use $j\defeq \sqrt{-1}$ to denote the imaginary unit. For a complex number $x=a+jb\in \mbb{C}$ where $a,b\in \mbb{R}$, $\bar{x}\defeq a-jb$ denotes its complex conjugate. We use $\Pp(E)$ to denote the probability of the event $E$. Vectors and matrices are written in lowercase boldface and uppercase boldface, respectively. Random variables are written in Sans Serif font. We use $\times$ to denote the Cartesian product and use $\oplus$ and $\odot$ to denote the addition and multiplication operations, respectively, over a finite field where the field size can be understood from the context if it is not specified.

\section{Problem Statement}\label{sec:problem}
The network considered in this paper is the compute-and-forward relay network introduced by Nazer and Gastpar in \cite{nazer2011CF}. Consider a $K$ source nodes $M$ destination nodes AWGN network as shown in Fig~\ref{fig:CF_model}. Each source node has a message $w_k\in\{1,2,\ldots,W\}$, $k\in\{1,\ldots,K\}$ which can alternatively be expressed by a length-$N'$ vector over some finite field, i.e., $\mathbf{w}_k\in\mbb{F}_p^{N'}$ with $W=p^{N'}$. This message is fed into an encoder $\mathcal{E}^N_k$ whose output is a length-$N$ codeword $\mathbf{x}_k\in\mbb{C}^N$. Each codeword is subject to a power constraint given by
\begin{equation}
    \frac{1}{N}\| \mathbf{x}_k\|^2 = \frac{1}{N}\sum_{n=1}^N |x_k[n]|^2 \leq P.
\end{equation}

\begin{figure}
    \centering
    \includegraphics[width=3.5in]{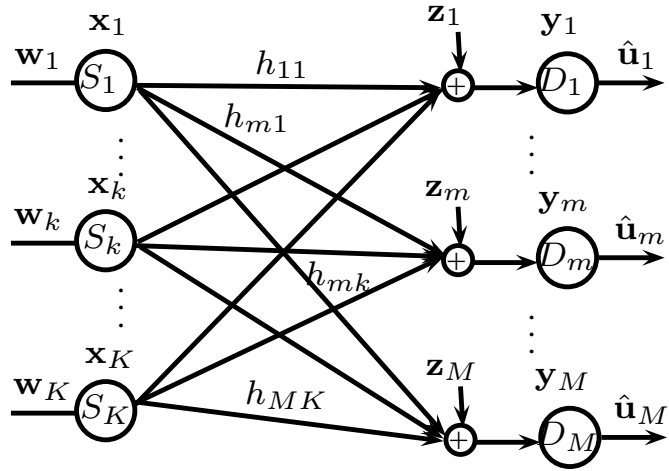}
    \caption{A compute-and-forward relay network where $S_1,\ldots,S_K$ are source nodes and $D_1,\ldots,D_M$ are destination nodes.}
    \label{fig:CF_model}
\end{figure}

The signal observed at destination $m$ is given by
\begin{equation}
    y_m[n] = \sum_{k=1}^K h_{mk}x_k[n] + z_m[n],
\end{equation}
where $h_{mk}\in\mbb{C}$ is the channel coefficient between the source node $k$ and destination $m$, and $z_m[n]\sim \mathcal{CN}(0,1)$. One can also define the channel model for using the channel $N$ times as
\begin{equation}\label{eqn:y_m}
    \mathbf{y}_m = \sum_{k=1}^K h_{mk}\mathbf{x}_k + \mathbf{z}_m.
\end{equation}
Instead of individual messages, each destination node is only interested in computing a function of messages
\begin{equation}
    \mathbf{u}_m = f_m(\mathbf{w}_1, \ldots, \mathbf{w}_K).
\end{equation}
Upon observing $\mathbf{y}_m$, the destination node $m$ forms $\hat{\mathbf{u}}_m = \mathcal{G}^N_m(\mathbf{y}_m)$ an estimate of $\mathbf{u}_m$. These functions are then forwarded to the central destination which can recover all the messages given sufficiently many functions.

\begin{define}[Computation codes]
    For a given set of functions $f_1, \ldots, f_M$, a $(N,N')$ computation code consists of a sequence of encoding/decoding functions $(\mathcal{E}^N_1, \ldots, \mathcal{E}^N_K)/(\mathcal{G}^N_1, \ldots, \mathcal{G}^N_M)$ described above and an error probability given by
    \begin{equation}
        P_{e,m}^{(N)} \triangleq \Pp\left(\left\{\hat{\mathbf{u}}_m \neq \mathbf{u}_m\right\}\right).
    \end{equation}
\end{define}

\begin{define}[Computation rate at the relay $m$]
    For a given channel vector $\mathbf{h}_m \triangleq [h_{m1},\ldots,h_{mK}]^T$ and a given function $f_m$, a computation rate $R(\mathbf{h}_m,f_m)$ is achievable at relay $m$ if for any $\varepsilon>0$ there is an $(N,N')$ computation code such that
    \begin{equation}
        N'\geq NR(\mathbf{h}_m,f_m)/\log(p) \text{~and~} P_{e,m}^{(N)}\leq \varepsilon.
    \end{equation}
    Note that the first condition is equivalent to saying that $W\geq 2^{N R(\mathbf{h}_m,f_m)}$.
\end{define}

In practice, since no cooperation among relays are assumed, a greedy protocol which mimics the behavior of random linear network coding is considered in \cite{nazer2011CF} where each relay computes and forwards the function with the highest computation rate. After that, if given those functions, the central destination is able to recover all the messages, the decoding is successful. Otherwise, the central destination declares failure. The achievable computation rate for the transmitters is then equal to $\min_m R(\mathbf{h}_m,f_m)$. Note that here we consider the case when all the transmitters transmit at a same rate for the sake of simplicity; however, a more general model where they can have different rates was considered in \cite{nazer2011CF}.

%



\section{Preliminaries}\label{sec:prelim}
It should be noted that using the theory in Diophantine approximation, Niesen and Whiting have shown in \cite{niesen11} that the lattice-based compute-and-forward described above is in general very inefficient in terms of degrees of freedom and a coding scheme relying on the channel state information at transmitters has been devised to achieve the full degrees of freedom. Regardless of this deficiency in the asymptotic regime, this paper considers the lattice-based compute-and-forward as it is so far one of the best schemes to exploit the structural gains in the finite SNR regime. Besides, lattice-based schemes are based on a more realistic assumption that channel state information is available only at receivers.

In this section, we briefly summarize background knowledge on lattices and nested lattice codes followed by some preliminaries on abstract algebra. For more details about lattices, lattice codes, and nested lattice codes, the reader is referred to \cite{erez04} \cite{erez05} \cite{conway1999sphere}. We then summarize the lattice-based compute-and-forward paradigm and the main result in \cite{nazer2011CF} and briefly mention the separation-based framework proposed in \cite[Section V]{Engin12}.

\subsection{Lattices}
An $N$-dimensional lattice $\Lambda$ is a discrete subgroup of $\mathbb{R}^N$ which closes under reflection and ordinary vector addition operation. i.e., $\forall \lambda\in\Lambda$ $-\lambda\in\Lambda$, and  $\forall \lambda_1, \lambda_2\in \Lambda$ $\lambda_1 + \lambda_2 \in \Lambda^N$. Some important operations and notions for lattices are defined as follows.
\begin{define}[Lattice Quantizer]
For a $\mathbf{x}\in\mathbb{R}^N$, the nearest neighbor quantizer associated with $\Lambda$ is denoted as
\begin{equation}
    Q_{\Lambda}(\mathbf{x})=\lambda\in\Lambda;~\|\mathbf{x}-\lambda\|\leq\|\mathbf{x}-\lambda'\|~\forall\lambda'\in\Lambda,
\end{equation}
where $\| .\|$ represents the $L_2$-norm operation.
\end{define}

\begin{define}[Fundamental Voronoi Region]
The fundamental Voronoi region $\mathcal{V}_{\Lambda}$ is defined as
\begin{equation}
    \mathcal{V}_{\Lambda}=\{ \mathbf{x}: Q_{\Lambda}(\mathbf{x})=\mathbf{0} \}.
\end{equation}
\end{define}

\begin{define}[Modulo Operation]
The $\mod \Lambda$ operation returns the quantization error with respect to $\Lambda$ and is represented as
\begin{equation}
    \mathbf{x}\mod \Lambda = \mathbf{x}-Q_{\Lambda}(\mathbf{x}).
\end{equation}
\end{define}

The second moment of a lattice is defined as the average energy per dimension of a uniform probability distribution over $\mathcal{V}_{\Lambda}$ as
\begin{equation}
    \sigma^2(\Lambda) = \frac{1}{\text{Vol}(\mc{V}_{\Lambda})}\frac{1}{N}\int_{\mathcal{V}_{\Lambda}}\| \mathbf{x} \|^2 \mathrm{d}\mathbf{x},
\end{equation}
where $\text{Vol}(\mc{V}_{\Lambda})$ is the volume of $\mathcal{V}_{\Lambda}$. The normalized second moment of the lattice is then defined as
\begin{equation}
    G(\Lambda) = \frac{\sigma^2(\Lambda)}{\text{Vol}(\mc{V}_{\Lambda})^{2/N}},
\end{equation}
which is lower bounded by that of a sphere which asymptotically approaches $\frac{1}{2\pi\exp(1)}$ in the limit as $N\rightarrow \infty$. Note that $G(\Lambda)$ is invariant to scaling.

We now define two important notions of goodness for lattices.
\begin{define}[Goodness for MSE Quantization]
    We say that a sequence of lattices is asymptotically good for MSE quantization if
    \begin{equation}
        \underset{N\rightarrow\infty}{\lim} G(\Lambda) = \frac{1}{2\pi \exp(1)}.
    \end{equation}
\end{define}
Consider the unconstrained AWGN channel $\mathbf{Y}=\mathbf{X}+\mathbf{Z}$ where $\mathbf{X}$, $\mathbf{Y}$, and $\mathbf{Z}\sim \mathcal{N}(0,\eta^2\cdot I)$ represent the transmitted signal, the received signal, and the noise, respectively. Moreover, a lattice is adopted as input and there is no power constraint on $\mathbf{X}$ so that any lattice point can be sent.
\begin{define}[Poltyrev-Goodness (or Goodness for AWGN Channel Coding)]
    We say that a sequence of lattices is asymptotically Poltyrev-good if whenever
    \begin{equation}\label{eqn:poltyre_good}
        \eta^2<\frac{\text{Vol}(\mc{V}_{\Lambda})}{2\pi\exp(1)},
    \end{equation}
    the error probability of decoding $\mathbf{X}$ from $\mathbf{Y}$ can be made arbitrarily small.
\end{define}
Here, by Poltyrev-good lattices, we mean a sequence of lattices that approach the Poltyrev limit defined in \eqref{eqn:poltyre_good}. There is a stronger version of Poltyrev-goodness stating that the sequence of lattices achieves an error exponent lower bounded by the Poltyrev exponent \cite{poltyrev94}. However, the proof of achieving Poltyrev exponent is more involved and is not required to prove the main results in this paper. Hence, we do not pursue it in this paper. The interested reader is referred to \cite{poltyrev94} and \cite{erez04}.

\subsection{Algebra}
In this subsection, we provide some preliminaries that will be useful in explaining our results in the following sections. All the Lemmas are provided without proofs for the sake of brevity; however, their proofs can be found in standard textbooks of abstract algebra, see for example \cite{Hungerford74}.

We first recall some basic definitions for commutative rings where many of them are covered in \cite{Feng10} (for those not in \cite{Feng10}, the reader is referred to \cite{Hungerford74}). Let $\mc{R}$ be a commutative ring. Let $a, b\neq 0 \in\mc{R}$ but $ab = 0$, then $a$ and $b$ are \textit{zero divisors}. If $ab = ba = 1$, then we say $a$ is a \textit{unit}. Two elements $a, b\in\mc{R}$ are associates if $a$ can be written as the multiplication of a unit and $b$. A non-unit element $\phi\in\mc{R}$ is a prime if whenever $\phi$ divides $ab$ for some $a, b \in \mc{R}$, either $\phi$ divides $a$ or $\phi$ divides $b$. An \textit{integral domain} is a commutative ring with identity and no zero divisors. An additive subgroup $\mc{I}$ of $\mc{R}$ satisfying $ar\in\mc{I}$ for $a\in\mc{I}$ and $r\in\mc{R}$ is called an \textit{ideal} of $\mc{R}$. An ideal $\mc{I}$ of $\mc{R}$ is proper if $\mc{I}\neq\mc{R}$. An ideal generated by a singleton is called a \textit{principal ideal}. A \textit{principal ideal domain} (PID) is an integral domain in which every ideal is principal. Famous and important examples of PID include $\mbb{Z}$, $\Zi$ and $\Zw$. Let $a, b\in\mc{R}$ and $\mc{I}$ be an ideal of $\mc{R}$; then $a$ is congruent to $b$ \textit{modulo} $\mc{I}$ if $a-b\in\mc{I}$. The quotient ring $\mc{R}/\mc{I}$ of $\mc{R}$ by $\mc{I}$ is the ring with addition and multiplication defined as
\begin{align}
    (a+\mc{I})+(b+\mc{I}) &= (a+b)+\mc{I}, \text{~and} \\
    (a+\mc{I})\cdot(b+\mc{I}) &= (a\cdot b)+\mc{I}.
\end{align}

A proper ideal $\mc{P}$ of $\mc{R}$ is said to be a \textit{prime ideal} if for $a, b\in\mc{R}$ and $ab\in\mc{P}$, then either $a\in\mc{P}$ or $b\in\mc{P}$. Two ideals $\mc{I}_1$ and $\mc{I}_2$ of $\mc{R}$ are \textit{relatively prime} if
\begin{equation}
    \mc{R} = \mc{I}_1+\mc{I}_2 \triangleq \{a+b: a\in\mc{I}_1, b\in\mc{I}_2\}.
\end{equation}
If two ideals $\mc{I}_1$ and $\mc{I}_2$ are relatively prime, then $\mc{I}_1\mc{I}_2 = \mc{I}_1\cap\mc{I}_2$. A proper ideal $\mc{O}$ of $\mc{R}$ is said to be a \textit{maximal ideal} if $\mc{O}$ is not contained in any strictly larger proper ideal. It should be noted that every maximal ideal is also a prime ideal but the reverse may not be true. Let $\mc{R}_1, \mc{R}_2, \ldots, \mc{R}_L$ be a family of rings, the direct product of these rings, denoted by $\mc{R}_1\times \mc{R}_2\times \ldots \times\mc{R}_L$, is the direct product of the additive Abelian groups $\mc{R}_l$ equipped with multiplication defined by the \textit{componentwise} multiplication.

Let $\mc{R}_1$ and $\mc{R}_2$ be rings. A function $\sigma:\mc{R}_1\rightarrow \mc{R}_2$ is a \textit{ring homomorphism} if
\begin{align}
    \sigma(a + b) &= \sigma(a) \oplus \sigma(b) ~\forall a,b\in\mc{R}_1 \text{~and} \\
    \sigma(a\cdot b) &= \sigma(a)\odot \sigma(b),~\forall a,b\in\mc{R}_1.
\end{align}
A homomorphism is said to be \textit{isomorphism} if it is bijective. It is worth mentioning that for an ideal $\mc{I}$ there is a natural ring homomorphism $\mod\mc{I}:\mc{R}\rightarrow \mc{R}/\mc{I}$. A $\mc{R}$-module $\mc{N}$ over a ring $\mc{R}$ consists of an Abelian group ($\mc{N},+$) and an operation $\mc{R}\times \mc{N}\rightarrow \mc{N}$ which satisfies the same axioms as those for vector spaces. Let $\mc{N}_1$ and $\mc{N}_2$ be $\mc{R}$-modules. A function $\varphi:\mc{N}_1\rightarrow \mc{N}_2$ is a \textit{$\mc{R}$-module homomorphism} if
\begin{align}
    \varphi(a + b) &= \varphi(a) \oplus \varphi(b) ~\forall a,b\in\mc{N}_1 \text{~and} \\
    \varphi(r a) &= r \varphi(a),~\forall r\in\mc{R}, a\in\mc{N}_1.
\end{align}

We now present some lemmas which serve as the foundation of the paper.
\begin{lemma}\label{lma:PID}
    If $\mc{R}$ is a PID, then every non-zero prime ideal is maximal.
\end{lemma}

\begin{lemma}\label{lma:MAX}
    Let $\mc{I}$ be an ideal in a ring $\mc{R}$ with identity $1_{\mc{R}}\neq 0$. If $\mc{I}$ is maximal and $\mc{R}$ is commutative, then the quotient ring $\mc{R}/\mc{I}$ is isomorphic to a field.
\end{lemma}

\begin{lemma}[Chinese Remainder Theorem]\label{lma:CRT}
    Let $\mc{R}$ be a commutative ring, and $\mc{I}_1,\ldots,\mc{I}_n$ be ideals in $\mc{R}$, such that they are relatively prime. Then,
    \begin{equation}
        \mc{R}/\cap_{i=1}^n\mc{I}_i \cong \mc{R}/\mc{I}_1\times\ldots\times\mc{R}/\mc{I}_n.
    \end{equation}
\end{lemma}

\begin{example}\label{exp:ring_iso}
    Consider the PID $\mbb{Z}$ and one of its ideal $6\mbb{Z}$. Note that one can do the prime factorization $6=2\cdot 3$. Now since 2 and 3 are primes, $2\mbb{Z}$ and $3\mbb{Z}$ are prime ideals. Also, since $2\mbb{Z}+3\mbb{Z}=\mbb{Z}$, they are relatively prime. This implies that $2\cdot 3\mbb{Z} = 2\mbb{Z}\cap 3\mbb{Z}$. One has that
    \begin{align}
        \mbb{Z}_6 &\cong \mbb{Z}/6\mbb{Z} = \mbb{Z}/2\cdot 3 \mbb{Z} \nonumber \\
        &\overset{(a)}{=} \mbb{Z}/2\mbb{Z}\cap 3\mbb{Z} \nonumber \\
        &\overset{(b)}{\cong} \mbb{Z}/2\mbb{Z}\times\mbb{Z}/3\mbb{Z} \nonumber \\
        &\overset{(c)}{\cong} \mbb{F}_2\times \mbb{F}_3,
    \end{align}
    where (a) follows from that $2\mbb{Z}$ and $3\mbb{Z}$ are relatively prime, (b) follows from Chinese Remainder Theorem, and (c) is from Lemma~\ref{lma:MAX}. One isomorphism is given as follows,
    \begin{align}
        0&\leftrightarrow (0,0),~~1\leftrightarrow (1,1), \nonumber \\
        2&\leftrightarrow (0,2),~~3\leftrightarrow (1,0), \nonumber \\
        4&\leftrightarrow (0,1),~~5\leftrightarrow (1,2), \nonumber
    \end{align}
    and the multiplication is defined componentwise. One can easily see from this example that the product of two fields may not be a field. In this example, the product is isomorphic to $\mbb{Z}_6$ which is a ring but not a field.
\end{example}

We now introduce two important PIDs, namely the Eisenstein integers $\Zw$ and the Gaussian integers $\Zi$. The ring of Eisenstein integers $\Zw$ is the collection of complex numbers of the form $a+b\omega$ where $a,b\in\mbb{Z}$ and $\omega=-\frac{1}{2}+j\frac{\sqrt{3}}{2}$. The ring of Gaussian integers $\Zi$ is the collection of complex numbers of the form $a+bj$ where again $a,b\in\mbb{Z}$. Both $\Zw$ and $\Zi$ are PIDs.  The group of units (closed under multiplication) in $\Zw$ is $\{\pm 1,\pm \omega, \pm \omega^2\}$ and that in $\Zi$ is $\{\pm 1, \pm j\}$. An Eisenstein integer $\phi$ is an Eisenstein prime if and only if one of the following mutually exclusive conditions hold:
\begin{enumerate}
    \item $|\phi|^2=3$,
    \item $\phi$ is equal to the product of a unit and any rational prime congruent to $2\mod 3$,
    \item $|\phi|^2$ is any rational prime congruent to $1\mod 3$.
\end{enumerate}
This means that $3$ ramifies in $\Zw$. All the rational primes congruent to $2\mod 3$ stay inert in $\Zw$ and those congruent to $1\mod 3$ split into two distinct primes in $\Zw$. An Gaussian integer $\phi$ is an Gaussian prime if and only if one of the following mutually exclusive conditions hold:
\begin{enumerate}
    \item $|\phi|^2=2$,
    \item $\phi$ is equal to the product of a unit and any rational prime congruent to $3\mod 4$,
    \item $|\phi|^2$ is any rational prime congruent to $1\mod 4$.
\end{enumerate}
This means that $2$ ramifies in $\Zi$. All the rational primes congruent to $3\mod 4$ stay inert in $\Zi$ and those congruent to $1\mod 4$ split into two distinct primes in $\Zi$.

For those Eisenstein primes (Gaussian primes) $\phi$ with $|\phi|^2=\phi\cdot\bar{\phi}$ being rational primes congruent to $1\mod 3$ ($1\mod 4$), one can verify that $\phi$ and $\bar{\phi}$ are both Eisenstein primes (Gaussian primes) but they are not associates. Moreover, it has been shown in \cite{german} that for every $x\geq 7$, there exists a rational prime of this form between $x$ and $2x$. Thus, the choices of $\phi$ satisfying the above property are abundant. In the following sections, we will focus on the ring of Eisenstein integers for the sake of brevity but the schemes and the results for the ring of integers and the ring of Gaussian integers can be obtained in a straightforward fashion.



\subsection{Lattice-Based Compute-and-Forward in \cite{nazer2011CF}}
In \cite{nazer2011CF}, Nazer and Gastpar proposed a novel paradigm called compute-and-forward which exploits the algebraic structure of lattices. Using lattices for communication has a rich history in the literature. Typically, a lattice that is Poltyrev-good is required to guarantee reliable communication \cite{poltyrev94} \cite{loeliger97} \cite{forney2000}. In addition to the Poltyrev-goodness, shaping has to be taken into account in order to achieve the AWGN channel capacity. By carefully shaping the lattices with their sublattices, Erez and Zamir show that lattices can indeed achieve AWGN capacity with lattice decoding \cite{erez04}. Functional computation in physical layer with such lattices has been realized to asymptotically approach the capacity for the bidirectional relay networks in \cite{wilson07} \cite{wilson10} \cite{nam08} \cite{nam10}. The reader is referred to a tutorial paper \cite{nazer11tutorial} for more details about using lattices for the bidirectional relay channels. One of the main contribution of \cite{nazer2011CF} is to provide a means to harness interference when there is no channel state information at transmitters. In the sequel, we briefly summarize the main results and the coding scheme in \cite{nazer2011CF}.

In \cite{nazer2011CF}, the functions $f_m$ are chosen to be linear combinations of codewords with coefficients being integers $\mathbf{a}_m=[a_{m1},\ldots,a_{mK}]$. Hence, the functions are completely characterized by those coefficients and the achievable computation rates are written as $R(\mathbf{h}_m,\mathbf{a}_m)$. These integer combinations of codewords correspond to linear combinations of messages
\begin{equation}
    \mathbf{u}_m = b_{m1} \mathbf{w}_1\oplus\ldots\oplus b_{mK} \mathbf{w}_K.
\end{equation}

Each source node adopts an identical nested lattice code of Erez and Zamir \cite{erez04}. Specifically, let $(\Lambda_f,\Lambda_c)$ be two lattices such that $\Lambda_c$ is a sublattice of $\Lambda_f$, i.e., $\Lambda_c \subseteq \Lambda_f$, where $\Lambda_f$ is Poltyrev-good and $\Lambda_c$ is simultaneously good for MSE quantization and Poltyrev-good. Each source node uses $\Lambda_f\cap\mc{V}_{\Lambda_c}$ a set of minimum-energy coset representatives of the quotient group $\Lambda_f/\Lambda_c$ as codebook. The source node $k$ first bijectively maps its message $\mathbf{w}_k$ to a lattice codeword $\mathbf{t}_k\in\Lambda_f\cap\mc{V}_{\Lambda_c}$ and sends a dithered version
\begin{equation}
    \mathbf{x}_k = (\mathbf{t}_k -\mathbf{u}_k) \mod \Lambda_c.
\end{equation}

Given a Gaussian integer vectors $\mathbf{a}_m=[a_{m1},\ldots,a_{mK}]^T$, the relay $m$ scales the received signal by $\alpha_m$ and adds the dithers back to form
\begin{align}\label{eqn:NG_y}
    \mathbf{y}'_m &= \left(\alpha_m\mathbf{y}_m + \sum_{k=1}^K a_{mk}\mathbf{u}_k\right) \mod \Lambda_c \nonumber \\
    &= (\mathbf{t}_{eq,m} + \mathbf{z}_{eq,m}) \mod \Lambda_c,
\end{align}
where
\begin{equation}\label{eqn:NG_t_eq}
    \mathbf{t}_{eq,m}= \sum_{k=1}^K a_{mk}\mathbf{t}_{mk} \mod \Lambda_c,
\end{equation}
and
\begin{equation}\label{eqn:NG_z_eq}
    \mathbf{z}_{eq,m} = \left(\alpha_m\mathbf{z}_m + \sum_{k=1}^K (\alpha_m h_{mk}-a_{mk})\mathbf{x}_k\right).
\end{equation}
Due to the linearity of lattice codes, $\mathbf{t}_{eq,m}$ is a codeword in $\Lambda_f\cap\mc{V}_{\Lambda_c}$ and hence one can directly compute this function at the relay $m$. Moreover, note that the distribution of the equivalent noise $\mathbf{z}_{eq,m}$ is in general not Gaussian but would become Gaussian in the limit as $N\rightarrow \infty$ if $\Lambda_c$ is good for quantization due to the Gaussian approximation principle \cite{zamir96} \cite[Remark 5]{forney03}. This results in a computation rate given by
\begin{equation}
    R(\mathbf{h}_m,\mathbf{a}_m,\alpha_m) = \log^+\left(\frac{P}{|\alpha_m|^2+P\|\alpha_m\mathbf{h}_m-\mathbf{a}_m\|^2}\right),
\end{equation}
where $\log^+(.)\defeq \max\{0,\log(.)\}$. Intuitively speaking, one can arbitrarily rotate and scale the received signals by $\alpha_m$ such that the resulting channel coefficients would be arbitrarily close to the Gaussian integer vector $\mathbf{a}_m$ and hence make the second term in the denominator vanish. However, one might as well end up blowing up the noise which is the first term in the denominator. It turns out that the optimal choice of $\alpha_m$ is the MMSE estimator given by
\begin{equation}
    \alpha_{\text{MMSE},m} = \frac{P\mathbf{h}^*_m\mathbf{a}_m}{1+P\|\mathbf{h}_m\|^2}.
\end{equation}
Plugging the $\alpha_{\text{MMSE},m}$, one obtains the main result in \cite{nazer2011CF} as follows.
\begin{theorem}[Nazer-Gastpar]
    For given channel coefficients $\mathbf{h}_m$ and Gaussian integer vector $\mathbf{a}_m$, the following computation rate is achievable at the relay $m$.
    \begin{align}\label{eqn:com_rate_m}
        R(\mathbf{h}_m,\mathbf{a}_m) &= R(\mathbf{h}_m,\mathbf{a}_m,\alpha_{\text{MMSE},m}) \nonumber \\
        &= \log^+\left(\left(\|\mathbf{a}_m\|^2-\frac{P|\mathbf{h}_m^*\mathbf{a}_m|^2}{1+P\|\mathbf{h}_m\|^2}\right)^{-1}\right).
    \end{align}
\end{theorem}
After computing $\mathbf{t}_{eq,m}$, the relay $m$ can recover the function $\mathbf{u}_m = \bigoplus_{k=1}^K b_{mk}\mathbf{w}_k$ where $b_{mk}\defeq\sigma(a_{mk})$ with $\sigma$ being the ring homomorphism used in Construction A for generating the underlying lattice \cite{LeechSloane71} \cite{conway1999sphere}. At the central destination, one can invert the matrix $\mathbf{B}=[\mathbf{b}_1,\ldots,\mathbf{b}_M]$ to recover all the messages if the matrix is invertible.

\begin{remark}
    The coding scheme in \cite{nazer2011CF} in fact separately transmits signals in the real and the imaginary parts. However, we find it easier for us to describe the scheme by directly looking at the complex field and Gaussian integers. In fact, this has motivated the generalization of the compute-and-forward paradigm to the ring of Eisenstein integers in \cite{Engin12} where each element in $\mathbf{A}$ is chosen from $\Zw$ instead of $\Zi$.
\end{remark}


\subsection{Separation-Based Compute-and-Forward in \cite[Section V]{Engin12}}
As mentioned above, the ensemble of lattices considered in \cite{nazer2011CF} is based on the construction of Erez and Zamir \cite{erez04} and hence is infinitely-dimensional and simultaneously good for channel coding and good for quantization. However, due to the lack of efficient shaping techniques in practice, we consider a somewhat more practical framework called the separation-based compute-and-forward proposed in \cite[Section V]{Engin12}. This framework attempts to separate the design of channel coding and data modulation so that one can let the dimension of channel coding grow and design the shaping to be optimal in a small dimensional space (in spite of being suboptimal in the $N$ dimensional space).

The separation-based compute-and-forward is shown in Fig.~\ref{fig:sep_framework} and is briefly summarized in the following. Without loss of generality, we first assume the message at source $k$ to be a length-$N'$ vector over some finite field $\mbb{F}_p$ with $p$ to be determined later, i.e., $\mathbf{w}_k\in \mbb{F}_p^{N'}$. The channel coding employed by all the source nodes is restricted to be the same linear code $C$ over $\mbb{F}_p$ in order to ensure that linear combinations (over $\mbb{F}_p$) of codewords themselves are valid codewords. On the other hand, the constellation has to be carefully chosen so that one can still benefit from the structural gain offered by the compute-and-forward strategy. It turns out that the key condition for this is a ring homomorphism between the extended version (to infinite constellation) of the signal constellation and $\mbb{F}_p$, the field that channel coding is implemented.

\begin{figure}
    \centering
    \includegraphics[width=5in]{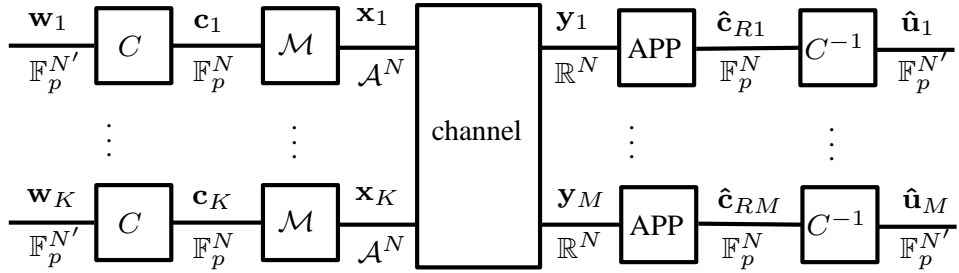}
    \caption{The separation-based compute-and-forward framework.}
    \label{fig:sep_framework}
\end{figure}

In \cite[Section V]{Engin12}, using fundamentals in commutative rings, the authors identify a family of signal constellations according to quotient rings of Eisenstein integers. Note that here and throughout, we will slightly abuse the notation and directly write the constellation as the quotient ring but it should be understood as a set of the minimum-energy coset representatives of that quotient rings. With this notation, the constellation proposed in \cite[Section V]{Engin12} is given by
\begin{equation}\label{eqn:S}
    \mc{A} \triangleq \Zw/\phi\Zw,
\end{equation}
where $|\phi|^2\triangleq p$ is a rational prime congruent to $1\mod 3$. The output is then scaled by $\gamma$ for satisfying the power constraint. Since $\Zw$ is a PID, from Lemma~\ref{lma:PID}, one has that $\phi\Zw$ is a maximal ideal. Also, the order of the quotient ring $\Zw/\phi\Zw$ is $|\Zw/\phi\Zw|=|\phi|^2=p$ (which is typically represented as $\left[\Zw:\phi\Zw\right]$ in the abstract algebra language). Hence, from Lemma~\ref{lma:MAX}, $\Zw/\phi\Zw\cong\mbb{F}_p$. i.e., the following ring isomorphism $\mc{M}$ exists,
\begin{equation}\label{eqn:homo}
    \Zw/\phi\Zw \overset{\mc{M}}{\underset{\mc{M}^{-1}}{\leftrightarrows}} \mbb{F}_p.
\end{equation}
Moreover, one can write $\Zw$ as the disjoint union of $p$ cosets of $\phi\Zw$ as follows,
\begin{equation}
    \Zw = \underset{a\in \mc{A}}{\bigcup} \left(\phi\Zw + a \right),
\end{equation}
where $\mc{A}$ is a set of minimum-energy coset representatives and $|\mc{A}|=p$. This induces a natural homomorphism from $\Zw$ to $\Zw/\phi\Zw$ via the $\mod \phi\Zw$ operation and hence $\sigma\defeq \mc{M}^{-1} \circ \mod\phi \Zw$ is a ring homomorphism described as follows,
\begin{equation}
    \sigma: \Zw \overset{\mod \phi \Zw}{\rightarrow} \Zw/\phi\Zw \overset{\mc{M}}{\underset{\mc{M}^{-1}}{\leftrightarrows}} \mbb{F}_p.
\end{equation}

The mapping from codeword elements to actual transmitted signals (before scaling) is then chosen to be this ring isomorphism $\mc{M}$. It has been shown in \cite{Feng10} \cite{Engin12} that the existence of such ring homomorphism $\sigma$ is crucial for exploiting the structural gains in compute-and-forward. Furthermore, this constellation provides other properties such as good shaping gain (in two-dimensional space) and good quantization of channel coefficients as $\Zw$ corresponds to hexagonal lattices. Upon receiving the signals, the receiver $m$ first computes the a posteriori probabilities (APP) for a given set of coefficients $[b_{m1},\ldots,b_{mK}]$ and then decodes to the codeword $\hat{\mathbf{c}}_{Rm}\in C$ that maximizes the APP. Note that since the encoders adopt the same linear code, one can then decode the corresponding $\hat{\mathbf{u}}_m$.

Unlike the framework considered in \cite{nazer2011CF} and \cite[Section III]{Engin12} in which infinitely-dimensional lattices are employed for channel coding and data modulation jointly, the separation approach allows one to let the dimension of channel coding grow while keeping the constellation size small so that optimal decoding is feasible. This also allows the use of well-developed codes on graphs (e.g., non-binary LDPC) for channel coding and enables one to employ iterative decoding such as message passing algorithm \cite{urbanke_book} to further reduce the decoding complexity. One key drawback of this scheme is that the channel coding has to work over $\mbb{F}_p$ for a constellation with $p$ elements. Hence, the decoding complexity increases dramatically as $p$ increases. This will be relaxed when the constellations proposed in Section~\ref{sec:const} are used together with the multilevel coding/multistage decoding proposed in Section~\ref{sec:encode_decode}.

\section{Proposed Product Construction of Lattices}\label{sec:prod_const}
Motivated by Theorem 2 in \cite{Feng10}, we propose the product construction of lattices shown in Fig.~\ref{fig:lattice_const}. Note that the proposed product construction can be used for generating lattices over $\mbb{Z}$, $\Zi$, and $\Zw$. In this section, we will only talk about $\mbb{Z}$ and $\Zw$ as the lattices over $\Zi$ can be obtained in a similar way as those over $\Zw$. The proposed lattices heavily rely on the existence of ring homomorphisms described in the following theorem.

\begin{theorem}
    Let $p_1, p_2,\ldots,p_L$ be a collection of distinct rational primes. There exists a ring isomorphism $\mc{M}:\times_{l=1}^L \mbb{F}_{p_l}\rightarrow \mbb{Z}/\Pi_{l=1}^L p_l \mbb{Z}$. Moreover,
    \begin{equation}
        \sigma:\mbb{Z}\overset{\mod \Pi_{l=1}^L p_l \mbb{Z}}\rightarrow\mbb{Z}/\Pi_{l=1}^L p_l \mbb{Z} \overset{\mc{M}}{\underset{\mc{M}^{-1}}{\leftrightarrows}} \mbb{F}_{p_1}\times\ldots\times\mbb{F}_{p_L},
    \end{equation}
    is a ring homomorphism. Similarly, let $\phi_1,\phi_2,\ldots,\phi_L$ be a collection of distinct Eisenstein primes that are relatively prime and with norm $|\phi_l|^2=q_l$ for $l\in\{1,\ldots,L\}$. There exists a ring isomorphism $\mc{M}:\times_{l=1}^L \mbb{F}_{q_l}\rightarrow \Zw/\Pi_{l=1}^L\phi_l \Zw $. Moreover,
    \begin{equation}
        \sigma:\Zw\overset{\mod \Pi_{l=1}^L\phi_l \Zw}\rightarrow\Zw/\Pi_{l=1}^L\phi_l \Zw \overset{\mc{M}}{\underset{\mc{M}^{-1}}{\leftrightarrows}} \mbb{F}_{q_1}\times\ldots\times\mbb{F}_{q_L},
    \end{equation}
    is a ring homomorphism.
\end{theorem}
\begin{proof}
    We only prove the theorem for $\mbb{Z}$. It follows that
    \begin{align}
        \mbb{Z}/\Pi_{l=1}^L p_l \mbb{Z} &\overset{(a)}{\cong} \mbb{Z}/\cap_{l=1}^L p_l\mbb{Z} \nonumber \\
        &\overset{(b)}{\cong} \mbb{Z}/p_1\mbb{Z}\times\ldots\times\mbb{Z}/p_L\mbb{Z} \nonumber \\
        &\overset{(c)}{\cong} \mbb{F}_{p_1}\times\ldots\times\mbb{F}_{p_L},
    \end{align}
    where (a) follows from that $p_l \mbb{Z}$ are relatively prime, (b) is from Chinese Remainder Theorem in Lemma~\ref{lma:CRT}, and (c) is due to the fact that $\mbb{Z}$ is a PID and Lemma~\ref{lma:MAX}. Therefore, the ring isomorphism $\mc{M}$ between the quotient ring $\mbb{Z}/\Pi_{l=1}^L p_l\mbb{Z}$ and the product of fields $\times_{l=1}^L \mbb{F}_{p_l}$ exists. Moreover, the modulo operation is a natural ring homomorphism; hence, $\sigma\defeq \mc{M}^{-1}\circ \mod \Pi_{l=1}^L p_l\mbb{Z}$ is a ring homomorphism.
\end{proof}
%

Throughout the paper, we will refer to a set of minimum-energy coset representatives of $\mbb{Z}/\Pi_{l=1}^L p_l \mbb{Z}$ ($\Zw/\Pi_{l=1}^L\phi_l\Zw$) as the \textit{signal constellation} or \textit{constellation} in short. Also, let $\mathcal{C}^l$, $l\in\{1,\ldots, L\}$, be the set of \textit{all} linear $(N,m^l)$ codes over $\mbb{F}_{p_l}$ ($\mbb{F}_{q_l}$) and $\mathcal{C}\triangleq \mathcal{C}^1\times\ldots\times \mathcal{C}^L$. i.e., $\mathcal{C}$ is the collection of all codes that can be represented as the Cartesian product of $L$ linear codes whose input lengths are $m^1,\ldots, m^2$, respectively, over $\mbb{F}_{p_l}$ ($\mbb{F}_{q_l}$). The construction consists of the following steps.
\begin{enumerate}
    \item Let $C = C^1\times\ldots \times C^L \in \mathcal{C}$ where $C^l \in \mathcal{C}^l$, $l\in\{1,\ldots,L\}$.
    \item Define $\Lambda^* \triangleq \mathcal{M}(C^1,\ldots,C^L)$ where for all the vectors $\mathbf{c}^1,\ldots,\mathbf{c}^L$ with equal length, $\mathcal{M}(\mathbf{c}^1,\ldots,\mathbf{c}^L)$ is defined as the elementwise mapping.
    \item Replicate $\Lambda^*$ over the entire $\mbb{R}^N$ ($\mbb{C}^N$) to form $\Lambda \triangleq  \Lambda^* + \Pi_{l=1}^L p_l \mbb{Z}^N$ ($\Lambda \triangleq \Lambda^* + \Pi_{l=1}^L \phi_l(\Zw)^N$).
\end{enumerate}
\begin{figure}
    \centering
    \includegraphics[width=4in]{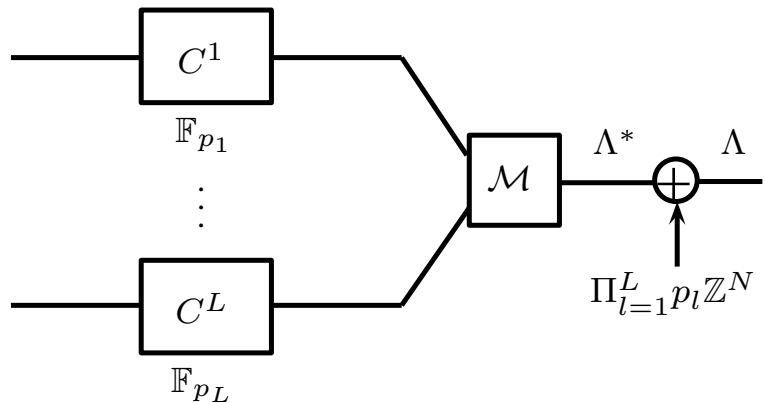}
    \caption{The proposed product construction of lattices.}
    \label{fig:lattice_const}
\end{figure}
Note that scaling by real (complex) numbers does not change the structure of a lattice; therefore, throughout the paper, we use $\Lambda \triangleq  \Lambda^* + \Pi_{l=1}^L p_l \mbb{Z}^N$ and $\Lambda \triangleq  \left(\Pi_{l=1}^L p_l\right)^{-1}\Lambda^* + \mbb{Z}^N$ interchangeably. For the lattices generated by the proposed product construction, we can show the following properties.
\begin{theorem}\label{thm:lattice}
    $\Lambda$ is a lattice. Moreover, there exists a sequence of such lattices that are simultaneously good for MSE quantization and Poltyrev-good under multistage decoding.
\end{theorem}
\begin{IEEEproof}
    See Appendix~\ref{apx:lattice}. The proof of the existence of Poltyrev-good lattices closely follows the proof by Forney in \cite{forney2000} instead of the Loeliger's proof in \cite{loeliger97}. The proof of the existence of lattices good for MSE quantization is a modification of a recent result by Ordentlich and Erez in \cite{ordentlich_erez_simple}.
\end{IEEEproof}

\begin{remark}
    When proving the Poltyrev-goodness, unlike Construction A lattices letting $p\rightarrow\infty$ and Construction D lattices letting $L\rightarrow\infty$, for the product construction lattices, we let $\Pi_{l=1}^L p_l\rightarrow \infty$ and allow one to play with these two parameters. Therefore, the proposed construction allows us to achieve the Poltyrev-limit with a significantly lower decoding complexity compared to Construction A lattices as now the complexity is not determined by the number of elements in $\Lambda^*$ but by the greatest divisor in the prime factorization of $|\Lambda^*|$. However, the complexity is higher than that of the Construction D lattices in \cite{BarnesSloane83} \cite{forney2000} whose complexity is always determined by coding over $\mbb{F}_2$. This is a direct consequence of that all primes should be distinct in the proposed product construction.
\end{remark}

\begin{remark}
    The possible value of $|\Lambda^*|$ is confined in a subset of $\mbb{N}$. For example, for $\mbb{Z}$, the proposed constellation allows $|\Lambda^*|$ to be any square-free integer \cite{sloaneOEIS}. Nonetheless, the choices of such $|\Lambda^*|$ are very rich and absorb Construction A lattices as special cases. The square-free integers are closely related to the M\"{o}bius function and can be identified efficiently without factorizing integers. The interested reader is referred to \cite{Auil13}.
\end{remark}

\subsection{Comparison with Construction D Lattices and its Variant}
At first glance, due to its multilevel nature, the proposed product construction looks similar to Construction D \cite{BarnesSloane83} \cite[Page 232]{conway1999sphere}. We compare and contrast the proposed product construction lattices and the Construction D lattices over $\Zw$ as it is more general. In order to make a detailed comparison, we first summarize Construction D extended to $\Zw$.

Let $\phi$ be an Eisenstein prime. From Lemma~\ref{lma:PID}, since $\Zw$ is a PID, $\phi \Zw$ is a prime ideal and hence a maximal ideal. If $|\phi|=q$ is a rational prime congruent to $2\mod 3$, from Lemma~\ref{lma:MAX}, we have that $\Zw/ \phi \Zw \cong \mbb{F}_{|\phi|^2} = \mbb{F}_{q^2}$. On the other hand, if $|\phi|^2 = q$ is a rational prime congruent to $1\mod 3$, again from Lemma~\ref{lma:MAX}, $\Zw/ \phi \Zw \cong \mbb{F}_{|\phi|^2}=\mbb{F}_q$. Therefore, in either case, we have a ring isomorphism $\mc{M}$ from $\mbb{F}_{|\phi|^2}$ to $\Zw$. This ring isomorphism will later be used for mapping codewords to constellations.

We first construct a set of nested linear codes $C^1 \subseteq C^2 \subseteq  \ldots \subseteq C^{r+1}$ over $\mbb{F}_{|\phi|^2}$ where $C^{r+1}$ is the trivial $(N,N)$-code and $C^l$ is a $(N,m^l)$-code for $l\in\{1,2,\ldots r\}$ with $m^1 \leq \ldots \leq m^r$. The codes are guaranteed to be nested by choosing $\{ \mathbf{g}_1,\ldots,\mathbf{g}_N \}$ which spans $C^{r+1}$ and then using the first $m^l$ vectors $\{ \mathbf{g}_1,\ldots,\mathbf{g}_{m^l} \}$ to generate $C^l$.

We are now ready to state the extended Construction D.

{\bf \underline{Construction D}} A lattice $\Lambda_{\text{D}}$ generated by the extended Construction D over $\Zw$ with $r+1$ level is given as follows.
\begin{equation}
    \Lambda_{\text{D}} = \bigcup \left\{ \phi^r (\Zw)^N + \sum_{1\leq l \leq r} \sum_{1\leq i \leq m^l} \phi^{l-1} \mc{M}(a_{li})\mc{M}(\mathbf{g}_i) |a_{li}\in\mbb{F}_{|\phi|^2}  \right\},
\end{equation}
where all the operations are over $\mbb{C}$.

A variant of Construction D called Construction by Code Formula has attracted a lot of attention since its introduction by Forney in \cite{forney88}, see for example \cite{harshan12} \cite{YanLingWu13} \cite{KosiOggier13}. Here, we also provide an extension of Construction by Code Formula to the complex field $\mbb{C}$. It is known that Construction by Code Formula does not always produce a lattice and it has been shown very recently in \cite{KosiOggier13} that one requires the nested linear codes closed under Schur product in order to have a lattice. Similar to Construction by Code Formula, the extended version does not always generate a lattice.

{\bf \underline{Construction by Code Formula}} Let $C^1 \subseteq C^2 \subseteq  \ldots \subseteq C^{r+1}$ be nested linear codes over $\mbb{F}_{|\phi|^2}$ as described above. A lattice $\Lambda_{\text{code}}$ generated by Construction by Code Formula over $\Zw$ with $r+1$ level is given as follows.
\begin{equation}
    \Lambda_{\text{code}} = \phi^r (\Zw)^N + \phi^{r-1} \mc{M}(C^r) + \ldots + \phi \mc{M}(C^{2}) + \mc{M}(C^{1}).
\end{equation}

Both Construction D and Construction by Code Formula admit an efficient (but suboptimal) decoding algorithm as follows. The decoder first reduces the received signal by modulo $\phi\Zw$. This will get rid of all the contribution from $C^2, \ldots, C^{r+1}$ and the remainder is a codeword from the linear code $C^1$. After successfully decoding, the decoder reconstructs and subtracts out the contribution from $C^1$ and divides the results by $\phi$. Now the signal becomes a noisy version (with variance $|\phi|^2$ times smaller than the original noise) of a lattice point from a lattice generated by the same construction with only $r$ level. So the decoder can then repeat the above procedure until all the codewords are decoded. In \cite{forney2000}, Forney \textit{et al.} show that Construction D lattices together with the above decoding procedure achieves the sphere bound and hence is Poltyrev-good.

One main difference between the proposed product construction and the two constructions described above is that the proposed product construction relies solely on the ring homomorphism while Construction D and Construction by Code Formula require the linear code at each level to be nested into those in the subsequent levels. In addition to this, another fundamental difference is that the proposed product construction allows the codes used in different levels to be over different fields while Construction D and Construction by Code Formula require them to be over the same field. Moreover, the mapping from $(\mbb{F}_{|\phi|^2}^N)^{r+1}$ to $(\Zw)^N$ as a whole used in the two constructions may not possess the ring homomorphism property as required by our product construction. i.e., sum of lattice points may not correspond to sum of codewords over $\mbb{F}_{|\phi|^2}$ for $C^1,\ldots,C^{r+1}$. The lack of ring homomorphisms renders these two constructions not straightforward to be used for compute-and-forward. This difference will be further discussed in Remark~\ref{rmk:D_revisit} in Section~\ref{sec:const2_sp}.

\section{Proposed Lattice-Based Multistage Compute-and-Forward}\label{sec:lattice_CF}
In this section, the proposed lattices are used to generate a sequence of nested lattice codes. Due to the multilevel nature, the proposed nested lattice codes admit multistage decoding and hence are computationally less complex. We then replace the nested lattice codes adopted in \cite{nazer2011CF} by the nested lattice codes generated by the proposed product construction and obtain similar achievable computation rates with multistage decoding. As corollaries, we also recover the main results in \cite{erez04} and \cite{wilson10} by our proposed lattices with multistage decoding.

\subsection{Main Result}
Here, we only consider the ring of integers $\mbb{Z}$ and the real channel coefficients, i.e., $h_{mk}\in\mbb{R}$. The results for the complex coefficients with either $\Zi$ or $\Zw$ can be obtained in a similar fashion. Let $p_1,\ldots,p_L$ be rational primes and $\mc{M}:\times_{l=1}^L\mbb{F}_{p_l}\rightarrow \mbb{Z}/\Pi_{l=1}^L p_l\mbb{Z}$ be the ring isomorphism. We note that each integer $a_{mk}\in\mbb{Z}$ can be represented as
\begin{equation}
    a_{mk} = \bar{a}_{mk} + \Pi_{l=1}^L p_l\tilde{a}_{mk},
\end{equation}
where $\bar{a}_{mk}\in\mbb{Z}/\Pi_{l=1}^L p_l \mbb{Z}, \tilde{a}_{mk}\in\mbb{Z}$. Moreover, each $\bar{a}_{mk}$ can be represented by its coordinate in $\times_{l=1}^L \mbb{F}_{p_l}$ as
\begin{equation}
    \bar{a}_{mk} = \mc{M}(b_{mk}^1,\ldots,b_{mk}^L).
\end{equation}
We can also write $\mathbf{a}_m = \bar{\mathbf{a}}_m + q\tilde{\mathbf{a}}_m$ where $\bar{\mathbf{a}}_m=\mc{M}(\mathbf{b}_m^1,\ldots,\mathbf{b}_m^L)$. In our proposed scheme, each transmitter decomposes its message $\mathbf{w}_k$ into $L$ sub-messages $\mathbf{w}^l_k$ over $\mbb{F}_{p_l}$ for $l\in\{1,\ldots,L\}$. The functions we aim to compute and the relay $m$ are given by
\begin{equation}
    \mathbf{u}_m^l \defeq b_{m1}^l\odot\mathbf{w}_{1}^l\oplus\ldots\oplus b_{mK}^l\odot\mathbf{w}_{K}^l,
\end{equation}
for $l\in\{1,\ldots,L\}$. We are now ready to state the main result of this section.


\begin{theorem}
    For given channel coefficients $\mathbf{h}_m$ and integer vector $\mathbf{a}_m$, the computation rate $R(\mathbf{h}_m,\mathbf{a}_m)$ described in \eqref{eqn:com_rate_m} is achievable under multistage decoding at the relay $m$.
\end{theorem}
Now suppose that we only have one transmitter and one receiver and the channel coefficient is $1$. The channel reduces to the point to point AWGN channel and the above theorem recovers the main results in \cite{erez04} with multistage decoding.
\begin{corollary}
    For the point to point AWGN channel, the proposed nested lattice codes together with multistage decoding achieves a rate of, $\frac{1}{2}\log\left(1+P\right)$, the channel capacity.
\end{corollary}
Also, consider the case we have two transmitters and one receiver with $h_{11}=h_{12}=1$. The channel model reduces to the first phase (MAC phase) of the bidirectional relay channel and when setting $a_{11}=a_{12}=1$, the above theorem recovers the main results in \cite{wilson10} with multistage decoding.
\begin{corollary}
    For the first phase of the two-way relay channel, the proposed nested lattice codes together with multistage decoding achieves $\frac{1}{2}\log\left(\frac{1}{2}+P\right)$, which is asymptotically optimal in the high SNR regime.
\end{corollary}

\subsection{Proof of the Main Result}
We consider an ensemble of nested lattice codes that can be regarded as a generalization of the ensemble in \cite{ordentlich_erez_simple} rather than the frequently used one by Erez and Zamir in \cite{erez04}. We first generate pairs of nested linear codes $(\mc{C}^l_f,\mc{C}^l_c)$ such that $\mc{C}^l_c\subseteq \mc{C}^l_f$ for $l\in\{1,\ldots,L\}$ as follows,
\begin{align}
    \mc{C}^l_c &= \{\mathbf{G}^l_c\odot\mathbf{w}^l | \mathbf{w}^l\in\mbb{F}_{p_l}^{m^l_c} \}, \\
    \mc{C}^l_f &= \{\mathbf{G}^l_f\odot\mathbf{w}^l | \mathbf{w}^l\in\mbb{F}_{p_l}^{m^l_f} \},
\end{align}
where $\mathbf{G}^l_c$ is a $N\times m^l_c$ matrix and
\begin{equation}
    \mathbf{G}^l_f = \begin{bmatrix}
                       \mathbf{G}^l_c & \mathbf{\tilde{G}}^l \\
                     \end{bmatrix},
\end{equation}
where $\mathbf{\tilde{G}}^l$ is a $N\times (m^l_f-m^l_c)$ matrix. We then generate (scaled) lattices $\Lambda_f$ and $\Lambda_c$ from the proposed product construction with the linear codes $\mc{C}^l_f$ and $\mc{C}^l_c$, respectively, as follows.
\begin{align}
    \Lambda_f &\triangleq  \gamma\left(\Pi_{l=1}^L p_l\right)^{-1} \mc{M}(\mc{C}^1_f,\ldots,\mc{C}^L_f) + \gamma\mbb{Z}^N, \nonumber \\
    \Lambda_c &\triangleq  \gamma\left(\Pi_{l=1}^L p_l\right)^{-1} \mc{M}(\mc{C}^1_c,\ldots,\mc{C}^L_c) + \gamma\mbb{Z}^N,
\end{align}
where $\gamma = 2\sqrt{NP}$. Clearly, $\Lambda_c\subseteq\Lambda_f$ and the design rate is given by
\begin{equation}
    R_{\text{design}} = \sum_{l=1}^L\frac{m^l_f-m^l_c}{N}\log(p_l).
\end{equation}
The design rate becomes the actual rate if every $\mathbf{G}^l_f$ is full-rank which will be fulfilled with high probability. Moreover, as shown in Appendix~\ref{apx:lattice} setting
\begin{equation}
    \sum_{l=1}^L\frac{m^l_c}{N}\log(p_l)\rightarrow \frac{1}{2}\left(\frac{4}{V_N^{2/N}}\right),
\end{equation}
ensures that the second moment converges to $P$ with high probability and the coarse lattice is good for MSE quantization. Besides, at the level $l$ for $l\in\{1,\ldots,L\}$, randomly choosing $\mathbf{G}^l_f$ would result in a capacity-achieving linear code with high probability, which in turn gives us a Poltyrev-good lattice. Therefore, in the following, we can assume that the coarse lattice is good for MSE quantization and the fine lattice is Poltyrev-good.

The transmitter $k$ first decomposes the message $\mathbf{w}_k$ into $(\mathbf{w}^1_k,\ldots,\mathbf{w}^L_k)$, where $\mathbf{w}^l_k$ is a length $(m_f^l-m_c^l)$ vector over $\mbb{F}_{p_l}$, and bijectively maps it to a lattice point $\mathbf{t}_k\in\Lambda_f\cap\mc{V}_{\Lambda_c}$ where
\begin{equation}
    \mathbf{t}_k = \left(\gamma\left(\Pi_{l=1}^L p_l\right)^{-1}\mc{M}(\mathbf{c}^1_k,\ldots,\mathbf{c}^L_k) + \gamma\zeta_k\right) \mod\Lambda_c,
\end{equation}
with $\zeta_k\in\mbb{Z}^N$ and $\mathbf{c}^l_k\defeq\mathbf{G}_f^l\odot[\mathbf{0}_{m_c^l}~\mathbf{w}^l_k]^T$. It then sends a dithered version
\begin{equation}
    \mathbf{x}_k = (\mathbf{t}_k-\mathbf{u}_k)\mod\Lambda_c.
\end{equation}
Similar to \eqref{eqn:NG_y}, by scaling the receive signal by $\alpha$ and adding the dithers back, one obtains
\begin{equation}
    \mathbf{y'}_m = (\mathbf{t}_{eq,m} + \mathbf{z}_{eq,m}) \mod \Lambda_c,
\end{equation}
where $\mathbf{t}_{eq,m}$ and $\mathbf{z}_{eq,m}$ are as \eqref{eqn:NG_t_eq} and \eqref{eqn:NG_z_eq}, respectively. Moreover, with the relationship $a_{mk}=\bar{a}_{mk}+\tilde{a}_{mk}$, one can further rewrite
\begin{align}
    &\mathbf{t}_{eq,m} = \left( \sum_{k=1}^K (\bar{a}_{mk}+ \Pi_{l=1}^L p_l \tilde{a}_{mk}) \mathbf{t}_k \right) \mod \Lambda_c \nonumber \\
    &= \left( \gamma(\Pi_{l=1}^L p_l)^{-1} \sum_{k=1}^K \mc{M}(b_{mk}^1,\ldots,b_{mk}^L) \mc{M}(\mathbf{c}_k^1,\ldots,\mathbf{c}_k^L) + \gamma\tilde{a}_{mk}\mathbf{t}_k \right) \mod \Lambda_c \nonumber \\
    &\overset{(a)}{=} \hspace{-4pt}\left(\gamma(\Pi_{l=1}^L p_l)^{-1}\mc{M}\left(\bigoplus_{k=1}^K b_{mk}^1\odot\mathbf{c}_k^1,\ldots,\bigoplus_{k=1}^K b_{mk}^L\odot\mathbf{c}_k^L\right)+ \gamma \zeta_m \right) \mod \Lambda_c,
\end{align}
where $\zeta_m\in\mbb{Z}^N$ and (a) holds because $\mc{M}(.)$ is a ring isomorphism. One can then decode the fine lattice point corresponding to $\mathbf{t}_{eq,m}$ by decoding the equivalent codeword $\bigoplus_{k=1}^K b_{mk}^l\odot\mathbf{c}_k^l$ level by level. This in turn gives an estimate of $\mathbf{u}_m^l$ for $l\in\{1,\ldots,L\}$. From the Gaussian approximation principle of MSE quantization good lattices \cite{zamir96} \cite[Remark 5]{forney03}, the equivalent noise would become Gaussian in the limit as $N\rightarrow \infty$. Thus, asymptotically, the probability of error is guaranteed to vanish whenever $\text{Vol}(\Lambda_f)$ is (slightly) larger than the volume of the noise ball since $\Lambda_f$ is Poltyrev-good under multistage decoding. Precisely, similar to the proof of the existence of Poltyrev-good lattices in Appendix~\ref{apx:lattice}, one can show that the equivalent channel seen at each level with the noise $\mathbf{z}_{eq,m}$ is regular. Let $\mathbf{z}^*_{eq,m}$ be the i.i.d. Gaussian random vector having distribution $\mc{N}(0,\sigma^2_{eq,m})$. Using the regularity of the channel and the fact that $D(\msf{Z}_{eq,m}\|\msf{Z}^*_{eq,m}) = h(\msf{Z}_{eq,m})- \frac{N}{2}\log 2\pi\exp(1)\sigma^2_{eq,m}$, one can show by following \cite{forney2000} that the error probability can be made arbitrarily small whenever
\begin{equation}
    \text{Vol}(\Lambda_f)^{\frac{2}{N}}>2\pi\exp(1)\sigma^2_{eq,m}2^{\frac{2}{N}D(\msf{Z}_{eq,m}\|\msf{Z}^*_{eq,m})},
\end{equation}
which converges to $2\pi\exp(1)\sigma^2_{eq,m}$ in the limit as $N\rightarrow \infty$ if the coarse lattice is good for quantization \cite{zamir96}.

The computation rate per real dimension achieved by this scheme is given by
\begin{align}
    R &= \frac{1}{N}\log\left(\frac{\text{Vol}(\Lambda_c)}{\text{Vol}(\Lambda_f)}\right) \nonumber \\
    &= \frac{1}{N}\log (\text{Vol}(\Lambda_c)) - \frac{1}{N}\log(\text{Vol}(\Lambda_f)) \nonumber \\
    &= \frac{1}{2}\log\frac{P}{G(\Lambda_c)}-\frac{1}{2}\log2\pi\exp(1)\sigma_{eq,m}^2 \nonumber \\
    &\overset{(a)}{=} \frac{1}{2}\log \left(\frac{P}{\sigma_{eq,m}^2}\right) \nonumber \\
    &= \frac{1}{2}\log\left(\frac{P}{|\alpha_m|^2+P\|\alpha_m\mathbf{h}_m-\mathbf{a}_m\|^2} \right),
\end{align}
where (a) follows from that the coarse lattice is good for MSE quantization. When $N$ is sufficiently large, one can choose the parameters $N$, $p_l$, $m_c^l$, and $m_f^l$ such that the design rate is arbitrarily close to the achievable computation rate derived above. Moreover, when sending signals along real and imaginary parts independently or considering the ring of Gaussian integers, one recovers the same computation rates as \eqref{eqn:com_rate_m} per complex dimension. Also, when considering the ring of Eisenstein integers, one obtains the same result in \cite{Engin12}.

One example of the proposed multistage compute-and-forward is provided below.
\begin{example}
    Consider the case where we only have two source node and we only focus on the computation at one destination (relay). Consider the isomorphism in Example~\ref{exp:ring_iso}. Let $\mathbf{G}_f^1=[1,1]^T$, $\mathbf{G}_f^2=[1,2]^T$, and $\mathbf{G}_c^1=\mathbf{G}_c^2=\{\emptyset\}$. i.e., we directly transmit the fine lattice points in this example. Suppose the codewords are
    \begin{equation}
    \mathbf{c}_1^1=\left[
                     \begin{array}{c}
                       1 \\
                       1 \\
                     \end{array}
                   \right],~~
    \mathbf{c}_1^2=\left[
                     \begin{array}{c}
                       1 \\
                       2 \\
                     \end{array}
                   \right],~~
    \mathbf{c}_2^1\left[
                     \begin{array}{c}
                       0 \\
                       0 \\
                     \end{array}
                   \right],~~
    \mathbf{c}_2^2 = \left[
                     \begin{array}{c}
                       2 \\
                       1 \\
                     \end{array}
                   \right].
    \end{equation}
    Ignoring the scaling factor, one has
    \begin{align}
    \mathbf{x}_1 &= \mc{M}\left(\left[
                     \begin{array}{c}
                       1 \\
                       1 \\
                     \end{array}
                   \right], \left[
                     \begin{array}{c}
                       1 \\
                       2 \\
                     \end{array}
                   \right]\right)=\left[
                     \begin{array}{c}
                       1 \\
                       5 \\
                     \end{array}
                   \right], \nonumber \\
    \mathbf{x}_2&=\mc{M}\left(\left[
                     \begin{array}{c}
                       0 \\
                       0 \\
                     \end{array}
                   \right],
                   \left[
                     \begin{array}{c}
                       2 \\
                       1 \\
                     \end{array}
                   \right]\right)=
                   \left[
                     \begin{array}{c}
                       2 \\
                       4 \\
                     \end{array}
                   \right].
      \end{align}
      Further, let the channel gains be $h_1=3=\mc{M}(1,0)$, $h_2=4=\mc{M}(0,1)$, and assume that there is no channel noise. The receiver will observe
      \begin{equation}
        \left[
          \begin{array}{c}
            11 \\
            31 \\
          \end{array}
        \right] = \left[
          \begin{array}{c}
            5 \\
            1 \\
          \end{array}
        \right] + 6\cdot\left[
          \begin{array}{c}
            1 \\
            5 \\
          \end{array}
        \right]
        = \mc{M}\left(\left[
          \begin{array}{c}
            1 \\
            1 \\
          \end{array}
        \right],\left[
          \begin{array}{c}
            2 \\
            1 \\
          \end{array}
        \right]\right) + 6\cdot\left[
          \begin{array}{c}
            1 \\
            5 \\
          \end{array}
        \right]
    \end{equation}
    At his particular destination, the multistage compute-and-forward will then compute $[1,1]^T$ and $[2,1]^T$ which corresponds to $1\odot\mathbf{c}_1^1\oplus 0\odot\mathbf{c}_2^1$ and $0\odot\mathbf{c}_1^2\oplus 1\odot\mathbf{c}_2^2$, respectively.
\end{example}


\section{Proposed Constellations for Separation-Based Compute-and-Forward}\label{sec:const}
In this section, we propose two families of constellations for the separation-based compute-and-forward. The constellations in the first family are isomorphic to the corresponding extension fields and those in the second family are isomorphic to the Cartesian product of finite fields; thus, the existence of ring homomorphisms can be shown. For some cases, the existence of such ring homomorphisms enables the proposed constellations to be directly used for separation-based compute-and-forward. More importantly, for both cases, there exist $\mbb{Z}$-module homomorphisms which will be the foundation of the proposed multilevel coding/multistage decoding scheme proposed in Section~\ref{sec:encode_decode}. This section is concluded in Section~\ref{sec:genreal_result} by providing the general result that absorbs the two proposed families. The reason that we divide it into two families is because the fundamental theorems used for showing these results are different and also this will ease the exposition of the results in this section and the next one. It must be emphasized that the theory required for showing the existence of the homomorphisms has been developed in \cite[Theorem 1 and Theorem 2]{Feng10}. But these specific constructions are not explicitly proposed in \cite{Feng10} and the multilevel coding/multistage decoding in Section~\ref{sec:encode_decode} is not given in \cite{Feng10}. Also, it is worth noting that the proposed technique works equally well for $\mbb{Z}$, $\Zi$, and $\Zw$, but here we particularly focus on $\Zw$ because the constellations obtained from $\Zw$ usually provide the best shaping gain among these three PIDs.

\subsection{The First Proposed Family of Constellations}\label{sec:const1}
Let $\phi$ be an Eisenstein prime with $\phi$ being the product of a unit and a rational prime $q$ congruent to $2\mod 3$. Since $\Zw$ is a PID, from Lemma~\ref{lma:PID}, $\phi\Zw$ is a prime ideal and hence a maximal ideal. Also, the order of the quotient ring $\Zw /\phi\Zw$ is $|\Zw /\phi\Zw|=|\phi|^2= q^2$. From Lemma~\ref{lma:MAX}, one has that
\begin{equation}
    \Zw /\phi\Zw \cong \mbb{F}_{q^2}.
\end{equation}
Thus, the following ring homomorphism exists
\begin{equation}\label{eqn:homo_prod_Eis1}
    \sigma: \Zw \overset{\mod \phi\Zw}{\rightarrow} \Zw/\phi\Zw \overset{\mc{M}}{\underset{\mc{M}^{-1}}{\leftrightarrows}} \mbb{F}_{q^2}.
\end{equation}
In order to exploit the structural gain, one then has to design the mapping $\mc{M}:\mbb{F}_{q^2}\rightarrow \Zw/\phi\Zw$ such that it is a ring isomorphism.

\begin{example}\label{exp:ring_Eis}
One example of this construction with $q = 5$ is given in Fig.~\ref{fig:ring_homo_Eis} where the labeling is the ring homomorphism and the multiplication in $\mbb{F}_{25}$ is defined by the irreducible polynomial $x^2+2x+4$ over $\mbb{F}_5$.
\begin{figure}
    \centering
    \includegraphics[width=5in]{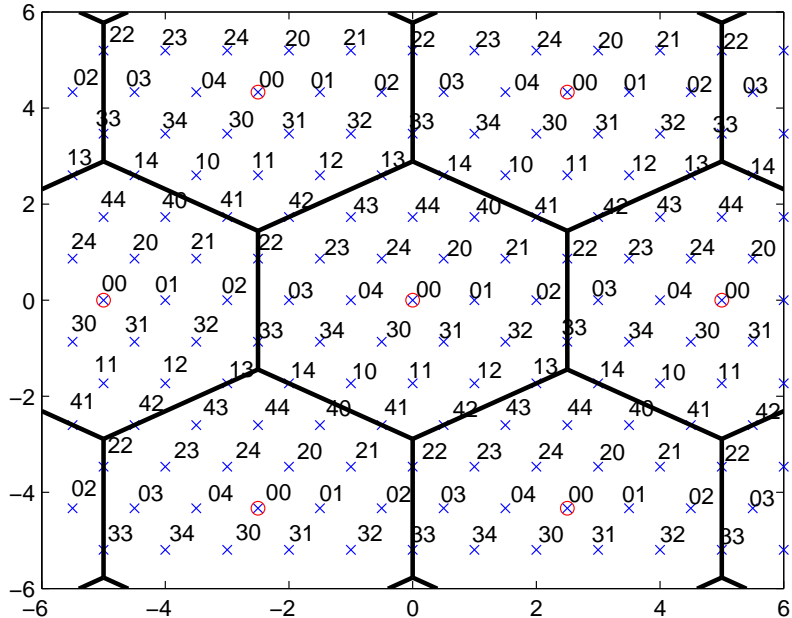}
    \caption{The proposed constellation with $q = 5$ and a ring homomorphism shown as the labeling with the irreducible polynomial $x^2+2x+4$.}
    \label{fig:ring_homo_Eis}
\end{figure}
\end{example}

%
%
One straightforward way to exploit this property is to use
\begin{equation}
    \mc{A}\triangleq \Zw/\phi\Zw,
\end{equation}
and scale the output by $\gamma$ for satisfying the power constraint, as signal constellation for the separation-based compute-and-forward framework in \cite{Engin12} and directly apply the ring isomorphism as signal mapping. And in fact, it can be shown that using Construction A over this family of constellations with appropriately chosen $q$, one would obtain a sequence of lattices that is simultaneously good for quantization and Poltyrev-good. This result is summarized in Appendix~\ref{apx:simul_good}.

Directly using constellations from this family implies that one has to work with a very large field $\mbb{F}_{q^2}$. This results in a significantly increased decoding complexity. It should be noted that the ring isomorphism between $\mbb{F}_{q^2}$ and $\Zw/\phi\Zw$ induces a $\mbb{Z}$-module isomorphism between $\mbb{Z}/q\mbb{Z}\times\mbb{Z}/q\mbb{Z}$ to $\Zw/\phi\Zw$. In the next section, we propose a novel encoding/decoding pair that incorporates the idea of multilevel coding and multistage decoding \cite{ImaiHirakawa77} \cite{WachsmannFischerHuber99} which largely relies on $\mbb{Z}$-module isomorphisms. The proposed encoding/decoding allows us to work over a potentially much smaller field $\mbb{F}_q$.


\subsection{The Second Proposed Family of Constellations}\label{sec:const2}
Let $\phi_1,\phi_2,\ldots,\phi_L$ be a collection of distinct Eisenstein primes with $|\phi_l|^2 = q_l,~\forall l\in\{1,2,\ldots,L\}$ congruent to $0\mod 3$ or $1\mod 3$. In addition, we also require $\phi_1,\phi_2,\ldots,\phi_L$ to be relatively prime. Since $\Zw$ is a PID, we have that each $\phi_l\Zw$ for $l\in\{1,2,\ldots,L\}$ is a maximal ideal. Moreover, we have the orders $|\Zw/\phi_l\Zw|=|\phi_l|^2=q_l$ for $l\in\{1,2,\ldots,L\}$. Therefore, one has that
\begin{align}\label{eqn:iso_prod_Eis2}
    \Zw/\Pi_{l=1}^L\phi_l\Zw &\cong \Zw/\cap_{l=1}^L\phi_l\Zw \nonumber \\
    &\overset{(a)}{\cong}\left(\Zw/\phi_1\Zw\right)\times \ldots\times\left(\Zw/\phi_L\Zw\right) \nonumber \\
    &\overset{(b)}{\cong}\mbb{F}_{q_1}\times\ldots\times\mbb{F}_{q_L},
\end{align}
where (a) follows from the Chinese Remainder Theorem in Lemma~\ref{lma:CRT} and (b) is from Lemma~\ref{lma:MAX}. Therefore, this implies that the following ring homomorphism exists,
\begin{equation}\label{eqn:homo_prod_Eis2}
    \sigma: \Zw \overset{\mod \Pi_{l=1}^L\phi_l\Zw}{\rightarrow} \Zw/\Pi_{l=1}^L\phi_l\Zw \overset{\mc{M}}{\underset{\mc{M}^{-1}}{\leftrightarrows}} \mbb{F}_{q_1}\times\ldots\times\mbb{F}_{q_L},
\end{equation}
where $\mc{M}$ is a ring isomorphism. We provide several examples as follows.
\begin{example}\label{exp:21pt}
    In this example, we choose $\phi_1 = 1+2\omega$ and $\phi_2 = 3+2\omega$ with $q_1 = 3$ and $q_2 = 7$, respectively. It can be verified that $\phi_1$ and $\phi_2$ are relatively prime. Then from Chinese Remainder Theorem, we have that $\Zw/\phi_1\phi_2\Zw\cong\mbb{F}_3\times\mbb{F}_7$. The constellation and a ring homomorphism from $\Zw$ to $\mbb{F}_{q_1}\times\ldots\times\mbb{F}_{q_L}$ are given in Fig.~\ref{fig:homo_21pt} where only the 21 points inside the big hexagon are used (ties can be broken arbitrarily) as constellation points.
    \begin{figure}
    \centering
    \includegraphics[width=5in]{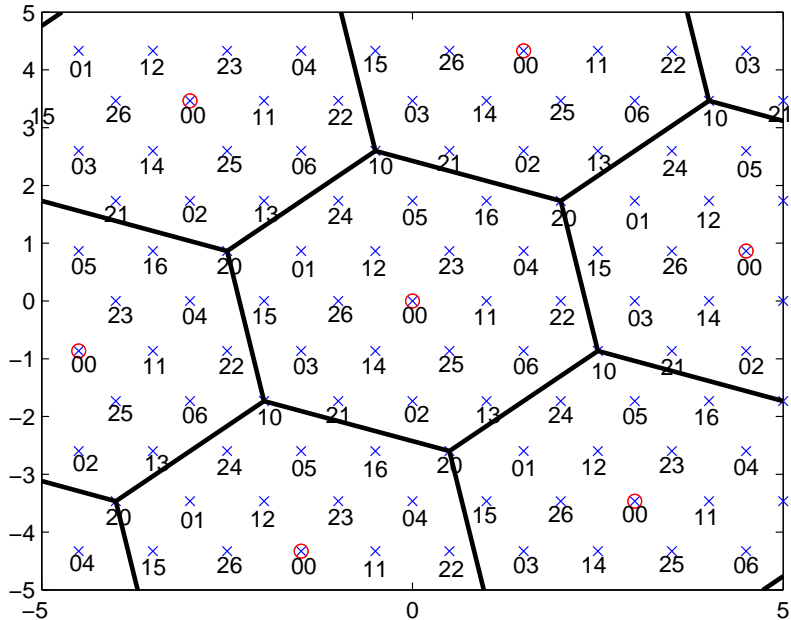}
    \caption{A 21-pt constellation in the second proposed family of constellations with $\phi_1 = 1+2\omega$ and $\phi_2 = 3+2\omega$.}
    \label{fig:homo_21pt}
    \end{figure}
    Let us provide a small example to illustrate how the compute-and-forward works according to this specific constellation and mapping. Suppose the source nodes have $c_1=(1,1)$ and $c_2=(2,6)$ in $\mbb{F}_3\times\mbb{F}_7$ and these will be mapped to $x_1=1$ and $x_2=-1$, respectively. Let the channel coefficients be $h_1=1$ and $h_2=\omega$ and there is no channel noise. In this case, the receiver will choose $a_1 = (1,1)$ and $a_2 = (1,2)$ since they are closest to $h_1$ and $h_2$, respectively. The received signal is then given by $h_1x_1+h_2x_2=1-\omega$ which corresponds exactly to the finite field result $(1,1)\odot(1,1)\oplus(1,2)\odot(2,6)=(0,6)$. One can also verify that this ring homomorphism induces a $\mbb{Z}$-module homomorphism $\varphi\defeq \mc{M}^{-1}\circ \mod\phi_1\phi_2\Zw$ where
    \begin{equation}\label{eqn:module_homo_21}
        \mc{M}(v^1,v^2) =  v^1(2\phi_2) + v^2 (3\phi_1) \mod \phi_1\phi_2\Zw,
    \end{equation}
    where $v^1\in\mbb{F}_3$ and $v^2\in\mbb{F}_7$.
\end{example}
Although the focus of this section is on $\Zw$, we also provide an example from $\Zi$.
\begin{example}
    Let $\phi_1 = 1+2j$ and $\phi_2 = 3 + 2j$. One has that $\Zi/\phi_1\phi_2\Zi\cong \mbb{F}_5 \times \mbb{F}_{13}$. This provides a constellation with 65 elements and the corresponding ring homomorphism shown in Fig.~\ref{fig:homo_65pt} where only the 65 points inside the fundamental Voronoi region are used.
    \begin{figure}
    \centering
    \includegraphics[width=5in]{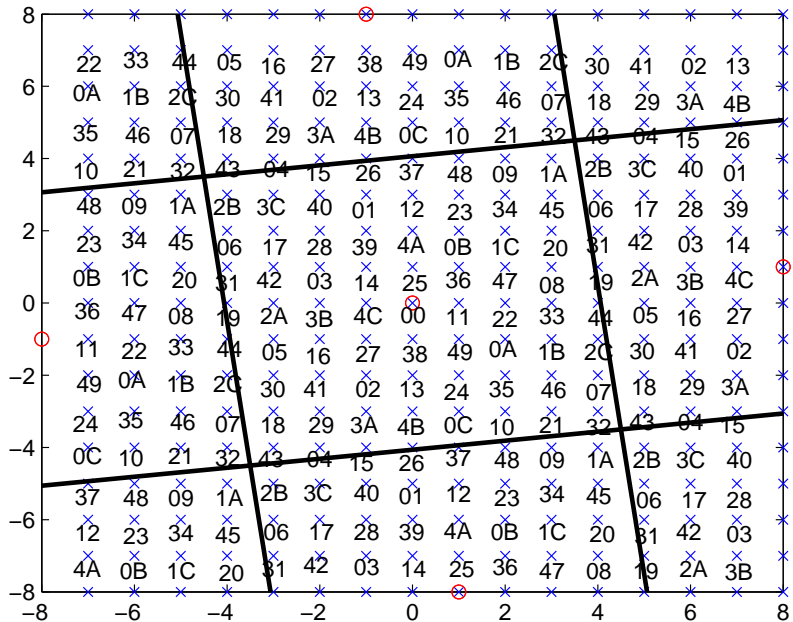}
    \caption{A 65-pt constellation in the second proposed family of constellations over $\Zi$ with $\phi_1 = 1+2j$ and $\phi_2 = 3+2j$. We use A, B, and C to denote $10, 11,$ and $12$ in $\mbb{F}_{13}$, respectively.}
    \label{fig:homo_65pt}
    \end{figure}
    One can verify that again, this ring homomorphism induces a $\mbb{Z}$-module homomorphism $\varphi\defeq \mc{M}^{-1}\circ \mod\phi_1\phi_2\Zw$ where
    \begin{equation}\label{eqn:module_homo_65}
        \mc{M}(v^1,v^2) =  v^1 (3\phi_2) + v^2 (6\phi_1) \mod \phi_1\phi_2 \Zi
    \end{equation}
    where $v^1\in\mbb{F}_5$ and $v^2\in\mbb{F}_{13}$.
\end{example}

Similar to the first proposed family of constellations, one can directly use the set of the coset representatives of $\Zw\rightarrow \Zw/\Pi_{l=1}^L\phi_l\Zw$ as constellations together with $\mc{M}$, which is a ring isomorphism, for signal mapping. But then, one may have to deal with coding over a large ring. Again, the decoding complexity will increase dramatically as the constellation size increases. Another way to take advantage of this family of constellations is to regard both sides of \eqref{eqn:iso_prod_Eis2} as finitely-generated Abelian groups, i.e., $\mbb{Z}$-modules. This implies that
\begin{equation}\label{eqn:module_iso_Eis2}
    \Zw/\Pi_{l=1}^L\phi_l\Zw \cong \mbb{Z}/q_1\mbb{Z}\times\ldots\times\mbb{Z}/q_L\mbb{Z},
\end{equation}
and the following $\mbb{Z}$-module homomorphisms exists
\begin{equation}\label{eqn:module_homo_prod_Eis2}
    \varphi: \Zw \overset{\mod \Pi_{l=1}^L\phi_l\Zw}{\rightarrow} \Zw/\Pi_{l=1}^L\phi_l\Zw \overset{\mc{M}}{\underset{\mc{M}^{-1}}{\leftrightarrows}} \mbb{Z}/q_1\mbb{Z}\times\ldots\times\mbb{Z}/q_L\mbb{Z},
\end{equation}
where now $\mc{M}$ is chosen to be a $\mbb{Z}$-module isomorphism whose existence is guaranteed by \eqref{eqn:module_iso_Eis2}. As mentioned before, the multilevel coding/multistage decoding that will be proposed later in Section~\ref{sec:encode_decode} only requires $\mbb{Z}$-module homomorphisms; hence, in the following, we focus on $\mbb{Z}$-module homomorphisms instead of ring homomorphisms. In the following theorem, we provide an explicit construction of $\mbb{Z}$-module isomorphisms for the proposed constellations.

\begin{theorem}
    Let $\phi_1,\ldots,\phi_L$ be a collection of Eisenstein primes with $|\phi_l|^2 = q_l,~\forall l\in\{1,2,\ldots,L\}$, congruent to $0\mod 3$ or $1\mod 3$. Also, $\phi_1,\ldots,\phi_L$ are relatively prime. The mapping $\mc{M}$ satisfying
    \begin{equation}\label{eqn:general_homo}
        \mc{M}(v^1,\ldots,v^L) \defeq \sum_{l=1}^L v^l\Pi_{l'=1,l'\neq l}^L \phi_{l'} \mod \Pi_{l=1}^L \phi_l\Zw,
    \end{equation}
    where $v^l\in\mbb{F}_{q_l}$, is a $\mbb{Z}$-module isomorphism from $\Zw/\Pi_{l=1}^L\Zw$ to $\mbb{F}_{q_1}\times\ldots,\times \mbb{F}_{q_L}$ and hence $\varphi\defeq \mc{M}^{-1}\circ \mod\Pi_{l=1}^L\Zw$ is a $\mbb{Z}$-module homomorphism.
\end{theorem}
\begin{IEEEproof}
    Let $v^l_k\in \mbb{F}_{q_l}$ for $k\in\{1,2\}$ and $l\in\{1,\ldots,L\}$. Consider
    \begin{equation}
        \mc{M}(v^1_k,\ldots,v^L_k) = \sum_{l=1}^L v^l_k\Pi_{l'=1,l'\neq l}^L \phi_{l'} \mod\Pi_{l=1}^L\Zw.
    \end{equation}
    One has that
    \begin{align}
        &\mc{M}(v^1_1,\ldots,v^L_1) + \mc{M}(v^1_2,\ldots,v^L_2) \mod\Pi_{l=1}^L\Zw \nonumber \\
        &= \sum_{l=1}^L (v^l_1+v^l_2)\Pi_{l'=1,l'\neq l}^L \phi_{l'} \mod\Pi_{l=1}^L\Zw \nonumber \\
        &= \sum_{l=1}^L (v^l_1\oplus v^l_2+\xi_l q_l)\Pi_{l'=1,l'\neq l}^L \phi_{l'} \mod\Pi_{l=1}^L\Zw \nonumber \\
        &= \sum_{l=1}^L (v^l_1\oplus v^l_2)\Pi_{l'=1,l'\neq l}^L \phi_{l'} + \sum_{l=1}^L \xi_l q_l\Pi_{l'=1,l'\neq l}^L \phi_{l'}  \mod\Pi_{l=1}^L\Zw \nonumber \\
        &\overset{(a)}{=} \sum_{l=1}^L (v^l_1\oplus v^l_2)\Pi_{l'=1,l'\neq l}^L \phi_{l'} + \sum_{l=1}^L \xi_l \bar{\phi}_l\Pi_{l'=1}^L \phi_{l'} \mod\Pi_{l=1}^L\Zw \nonumber\\
        &= \sum_{l=1}^L (v^l_1\oplus v^l_2)\Pi_{l'=1,l'\neq l}^L \phi_{l'} \mod\Pi_{l=1}^L\Zw \nonumber\\
        &= \mc{M}(v^1_1\oplus v^1_2,\ldots,v^L_1\oplus v^L_2),
    \end{align}
    where $\xi_l \in \mbb{Z}$ and the equality (a) is due to the fact that $q_l$ can be uniquely factorized (up to associates) as $\phi_l\bar{\phi}_l$.

\end{IEEEproof}

\begin{example}\label{exp:63pt}
    We now consider an example that has more than 2 levels. We choose $\phi_1 = 1-\omega$, $\phi_2 = 1-2\omega$, and $\phi_3 = 3+2\omega$ with $q_1 = 3$, $q_2 = 7$, and $q_3 =7$, respectively. It can be verified that $\phi_1$, $\phi_2$, and $\phi_3$ are relatively prime. Then from Chinese Remainder Theorem, we have that $\Zw/\phi_1\phi_2\phi_3\Zw\cong\mbb{F}_3\times\mbb{F}_7\times\mbb{F}_7$. This will give us a constellation with 147 elements shown in Fig.~\ref{fig:homo_147pt} where only the 147 points inside the big hexagon are used (ties can be broken arbitrarily) as constellation points. The $\mbb{Z}$-module homomorphism provided in \eqref{eqn:general_homo} are also plotted where one can verify that it is indeed a valid homomorphism.

    \begin{figure}
    \centering
    \includegraphics[width=5in]{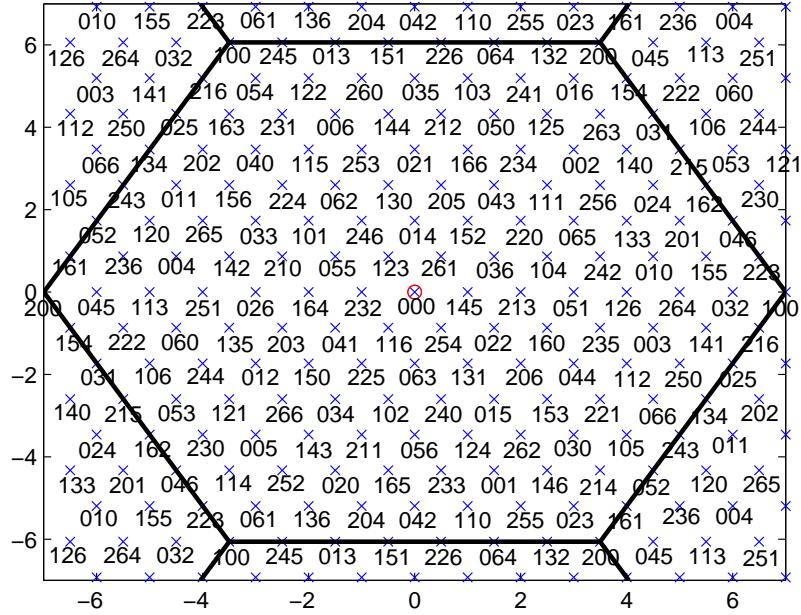}
    \caption{A 147-pt constellation in the second proposed family of constellations with $\phi_1 = 1-\omega$, $\phi_2 = 1-2\omega$, and $\phi_3 = 3+2\omega$.}
    \label{fig:homo_147pt}
    \end{figure}
\end{example}

\subsection{A Special Subclass of the Second Proposed Family of Constellations}\label{sec:const2_sp}
We now discuss an interesting subclass of the second proposed family of constellations. The constellations in this subclass have only two levels with equal size. Let $\phi_1=\phi$ and $\phi_2=\bar{\phi}$ with $|\phi|^2=|\bar{\phi}|^2=q$ be a rational prime congruent $1\mod 3$. As mentioned in Section~\ref{sec:prelim}, $\phi$ and $\bar{\phi}$ are relatively prime and the choices of such primes are abundant. Again, from the Chinese Remainder Theorem, one has that $\Zw/\phi\bar{\phi}\Zw\cong \mbb{F}_q^2$. One interesting feature of this subclass of constellations is that since two levels are over the same field $\mbb{F}_q$, $\mbb{Z}$-module homomorphisms can be generated by choosing any two linearly independent vectors as generators. Furthermore, as will be discussed later on, for this subclass of constellations, it is possible to include the idea of flexible decoding \cite{Brett11}. In what follows, we provide several interesting examples.

\begin{example}
Let $\phi=3+2\omega$ with $|\phi|^2=7$. Thus, we have $\Zw/\phi\bar{\phi}\Zw\cong \mbb{F}_7^2$. The constellation and a $\mbb{Z}$-module homomorphisms generated by
\begin{equation}\label{eqn:module_homo}
    \mc{M}(v^1,v^2) \triangleq v^1 + v^2 \omega \mod \phi\bar{\phi}\Zw.
\end{equation}
where $v^1, v^2\in\mbb{F}_q$, are shown in Fig.~\ref{fig:homo_49pt} where only the 49 points inside the big hexagon are used (ties can be broken arbitrarily) as constellation points.
    \begin{figure}
    \centering
    \includegraphics[width=5in]{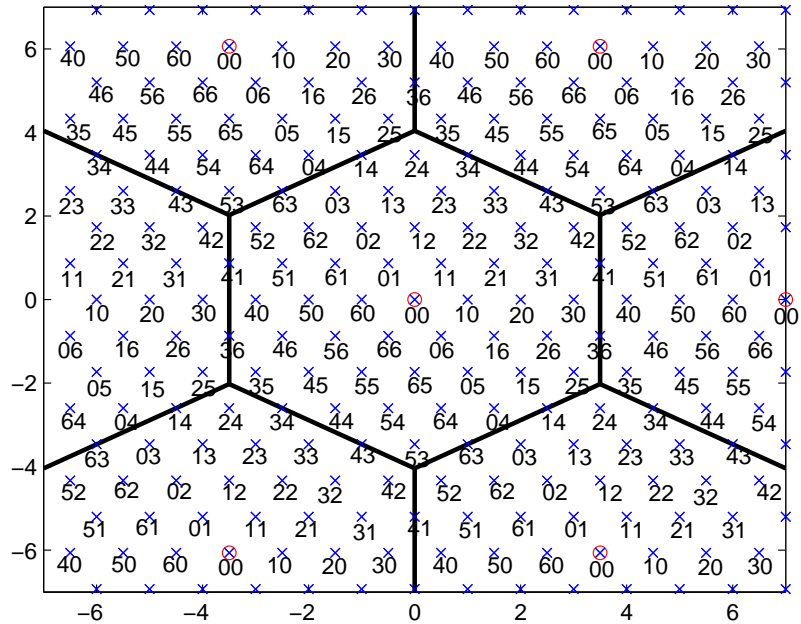}
    \caption{A 49-pt constellation in the second proposed family of constellations with $\phi = 3+2\omega$ and homomorphism defined in \eqref{eqn:module_homo}.}
    \label{fig:homo_49pt}
    \end{figure}
\end{example}

In general, there are many other ways to do the labeling. For example, one can incorporate the idea of Ungerboeck \cite{Ungerboeck82} where the minimum intra-subset Euclidean distance is maximized when partitioning at each level. However, this technique does not guarantee that the resulting mapping has the desired homomorphism property. One example of a signal mapping generated by the labeling technique of Ungerboeck is given in Fig.~\ref{fig:ungerboeck}. In each level, we use a ring isomorphism from $\mbb{F}_7$ to $\mbb{Z}$ to do the labeling. For example, if the first bit is set to be 0, one observes that all the points in $\mbb{C}$ corresponding to the seven points $(0,0), (0,1), \ldots, (0,6)$ in $\mbb{F}_7^2$ are mapped via a ring isomorphism. Similarly, inside each small hexagon, the mapping is done via a ring isomorphism. However, one can easily see that this mapping is in fact not a ring (or $\mbb{Z}$-module) isomorphism from $\mbb{F}_7^2$ to $\Zw$. Therefore, although it is a very powerful mapping technique in point-to-point communication, blindly applying the mapping of Ungerboeck usually provides significantly less rates as will be shown in Section~\ref{sec:simulation}.

We now show that it is possible to follow the guideline of Ungerboeck while maintaining the desired property. For the constellation generated by $\Zw/\phi\bar{\phi}\Zw$, one way to do this is to choose
\begin{equation}\label{eqn:module_homo_ungerboeck}
    \mc{M}(v^1,v^2) \triangleq v^1 + v^2 \phi \mod \phi\bar{\phi}\Zw,
\end{equation}
where $v^1,v^2\in\mbb{F}_q$. A labeling generated by this method can be found in Fig.~\ref{fig:ungerboeck2} where again the 49 points inside the big hexagon are used. In this example, one can verify that the intra-subset distance is indeed maximized. For example, the minimum distance of the set $(0,0), (0,1), \ldots, (0,6)$ is the largest one can get for this constellation. Furthermore, one can also verify that this mapping is indeed an homomorphism.

\begin{figure}
    \centering
    \includegraphics[width=5in]{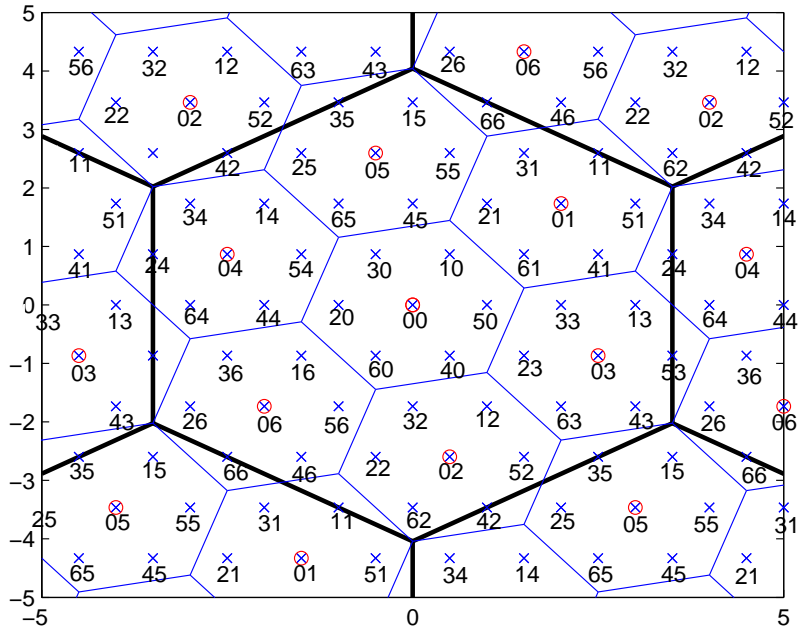}
    \caption{A 49-pt constellation with a labeling strategy obtained by blind apply Ungerboeck's idea \cite{Ungerboeck82}.}
    \label{fig:ungerboeck}
\end{figure}

\begin{figure}
    \centering
    \includegraphics[width=5in]{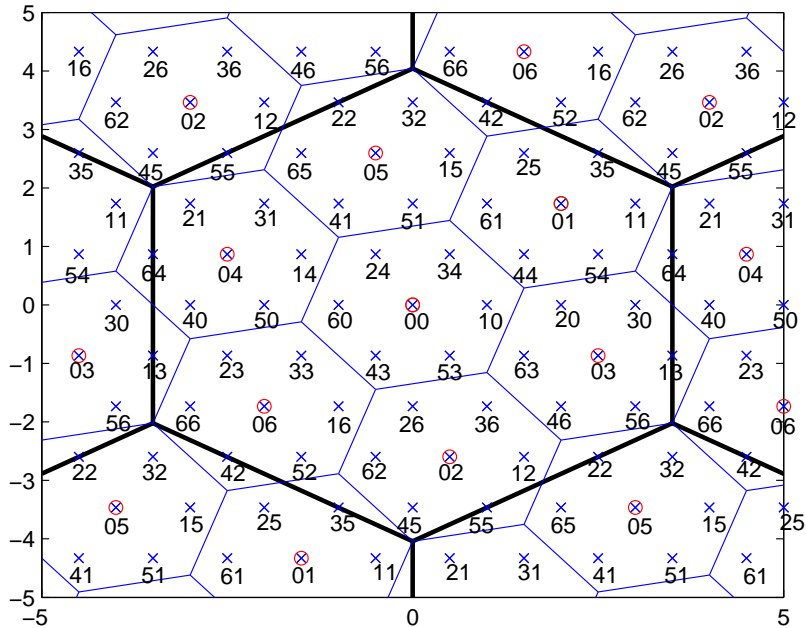}
    \caption{A 49-pt constellation with a labeling strategy obtained by \eqref{eqn:module_homo_ungerboeck}.}
    \label{fig:ungerboeck2}
\end{figure}

\begin{remark}[Extended Constructions Revisited]\label{rmk:D_revisit}
    One may have already noticed that the homomorphism we use in \eqref{eqn:module_homo_ungerboeck} resembles the mapping used in the extended Construction by Code Formula. Indeed, if we use the proposed product construction with this homomorphism, the lattice obtained would be
    \begin{equation}
        \Lambda_{\text{prod}} = q(\Zw)^N + \phi C^2 + C^1,
    \end{equation}
    and the extended Construction by Code Formula would generate
    \begin{equation}
        \Lambda_{\text{code}} = \phi^2(\Zw)^N + \phi C^2 + C^1,
    \end{equation}
    where $C^1$ and $C^2$ are $N$-dimensional linear codes over $\mbb{F}_q$. These two lattices look very similar. However, a subtle difference is that when we reduce $(\Zw)^N$ by modulo $\phi^2(\Zw)^N$, there is no guarantee that this mapping is a $\mbb{Z}$-module homomorphism. This is precisely why the extended Construction by Code Formula (and the extended Construction D as well) requires the linear codes to be nested in order to have a lattice. On the other hand, by Chinese Remainder Theorem, we have shown that when we reduce $(\Zw)^N$ by modulo $q(\Zw)^N$, $\mbb{Z}$-module homomorphisms exist and one can easily verify that \eqref{eqn:module_homo_ungerboeck} is indeed a valid one.
\end{remark}


\subsection{The General Result}\label{sec:genreal_result}
Here, we summarize in the following theorem the proposed constellations by providing the general result which absorbs all the proposed constellations and those in \cite{Engin12} as special cases. The proof is similar to those above and hence is omitted.
\begin{theorem}\label{thm:mix_prime}
    Let $\phi_l$ with $|\phi|^2=q_l$ congruent to $1\mod 3$ for $l\in\{1,\ldots,L\}$ and $\tilde{\phi}_{l'}$ with $|\tilde{\phi}_{l'}|=q_{l'}$ congruent to $2\mod 3$ for $l'\in\{1,\ldots,L'\}$ be a collection of distinct Eisenstein primes that are relatively prime. The following ring homomorphism exists
    \begin{equation}
        \sigma: \Zw\rightarrow \Zw/\Pi_{l=1}^L \Pi_{l'=1}^{L'} \phi_l\tilde{\phi}_{l'} \Zw \overset{\mc{M}}{\underset{\mc{M}^{-1}}{\leftrightarrows}} \times_{l=1}^L \mbb{F}_{q_l}\times \times_{l'=1}^{L'}\mbb{F}_{q_{l'}^2}.
    \end{equation}
    Moreover, when viewing the rings considered as finitely-generated Abelian groups, one has the following $\mbb{Z}$-module homomorphism
    \begin{equation}
        \varphi: \Zw\rightarrow \Zw/\Pi_{l=1}^L \Pi_{l'=1}^{L'} \phi_l\tilde{\phi}_{l'} \Zw \overset{\mc{M}}{\underset{\mc{M}^{-1}}{\leftrightarrows}} \times_{l=1}^L \mbb{Z}/q_l\mbb{Z}\times \times_{l'=1}^{L'}(\mbb{Z}/q_{l'}\mbb{Z})^2.
    \end{equation}
\end{theorem}
This theorem suggests that the constellation $\Zw/\Pi_{l=1}^L\phi_l \Zw$ is suitable for separation-based compute-and-forward for \textit{any} collection of Eisenstein primes that are relatively prime. Note that this result indicates that one can freely choose the primes regardless which congruence classes they belong to. In what follows, we provide one example of this kind.

\begin{example}
Let $\phi_1 = 1+2\omega$ and $\tilde{\phi}_1 = 2$. It should be noted that $|\phi_1|^2 = 3$ is congruent to $0\mod 3$ and $|\tilde{\phi}_1| = 2$ is congruent to $2\mod 3$. Theorem~\ref{thm:mix_prime} provides that $\Zw/\phi_1\tilde{\phi}_1\Zw\cong \mbb{F}_3\times\mbb{F}_{2^2}$. The constellation and the ring homomorphism is given in Fig.~\ref{fig:homo_12pt} where we use the first index to denote the element in $\mbb{F}_3$ and use the last two indices to denote the element in $\mbb{F}_{2^2}$. The multiplication is defined componentwise and the multiplication over $\mbb{F}_{2^2}$ is determined by the irreducible polynomial $x^2+x+1$. One can verify that the labeling indeed is a ring homomorphism. Moreover, by viewing $\mbb{F}_{2^2}$ as the finitely-generated Abelian group $(\mbb{Z}/2\mbb{Z})^2$, one obtains a $\mbb{Z}$-module homomorphism.
    \begin{figure}
        \centering
        \includegraphics[width=5in]{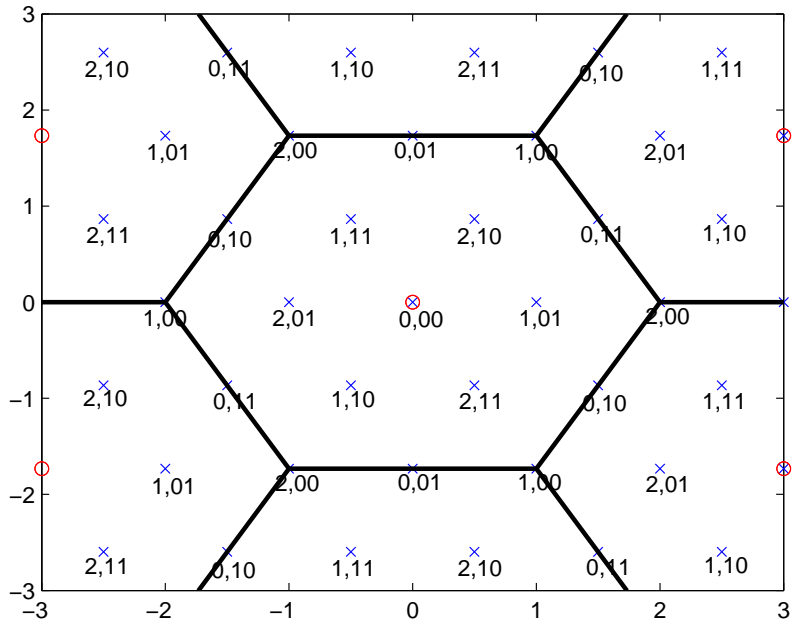}
        \caption{A 12-pt constellation generated by $\phi_1 = 1+2\omega$ and $\tilde{\phi}_1 = 2$.}
        \label{fig:homo_12pt}
    \end{figure}
\end{example}

\section{Proposed Separation-Based Multistage Compute-and-Forward}\label{sec:encode_decode}
In this section, we propose a multilevel encoding/multi-stage decoding scheme where only $\mbb{Z}$-module homomorphisms are required for exploiting the structural gains. For the proposed constellations, the proposed multilevel coding scheme allows one to significantly reduce the decoding complexity at the price of a slight rate reduction. Specifically, for the constellations with $\Pi_{l=1}^L q_l$ elements, the employment of the proposed multilevel coding/multistage decoding admits a low decoding complexity that is dominated by $\max_l{q_l}$ instead of $\Pi_{l=1}^L q_l$ the constellation size. Before starting, we note that the following description of the proposed multilevel coding/multistage decoding only considers the second proposed family of constellations. The other cases can be obtained straightforwardly.

\subsection{Encoding/Decoding}
Let $\phi_1,\phi_2,\ldots,\phi_L$ be a collection of distinct Eisenstein prime with $|\phi_l|^2=q_l$ a rational prime congruent to $0\mod 3$ or $1\mod 3$. Also, $\phi_1,\phi_2,\ldots,\phi_L$ are relatively prime. Let the messages be length-$N'$ vectors over $\mbb{F}_p$ such that
\begin{equation}
    p^{N'} \approx q_1^{m^1}\cdot q_2^{m^2}\cdot \ldots q_L^{m^L}.
\end{equation}
Each source node $S_k$ first splits its input stream $\mathbf{w}_k\in \mbb{F}_p^{N'}$ into $L$ streams, namely $\mathbf{w}_{k}^1\in\mbb{F}_{q_1}^{m^l}, \ldots, \mathbf{w}_{k}^L\in\mbb{F}_{q_L}^{m^L}$. For an Eisenstein prime $\phi$ such that $|\phi|=q$ is a rational prime congruent to $2\mod 3$, the encoder splits the input stream into two 2 streams which are over $\mbb{F}_q$. Let $C^l$ be linear code over $\mbb{F}_{q_l}$ adopted in level $l$ and let $G^l$ be the generator matrices for $l\in\{1,2,\ldots, L\}$. The rate of these linear codes are chosen to be $R_l=m^l/N \cdot \log(q_l)$ such that the outputs have the same length $N$. In total, the targeted computation rate of the proposed multilevel coding/multistage decoding scheme is $R_{\text{MLC}}=\sum_{l=1}^L R_l$. We individually encode each stream with the corresponding linear code as $\mathbf{c}_{k}^l = \mathbf{w}_{k}^l G^l$ for $l\in\{1,2,\ldots, L\}$. The encoder $k$ then takes $L$ of these symbols and maps them to a symbol from $\mc{A}$ via $\mc{M} \triangleq \gamma\cdot\varphi^{-1}$. The overall encoding process is summarized in Fig.~\ref{fig:encoder}.

\begin{figure}
    \centering
    \includegraphics[width=4in]{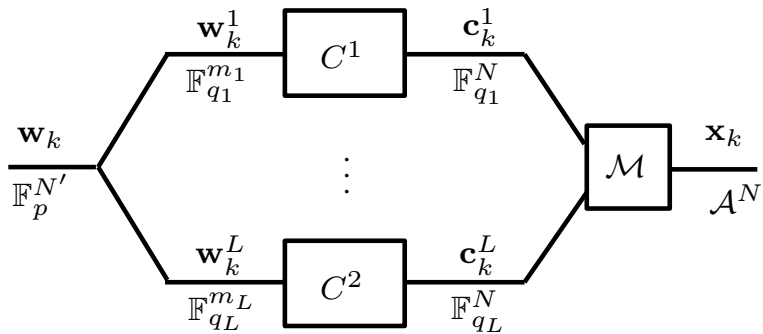}
    \caption{The encoder of the proposed multilevel coding scheme.}
    \label{fig:encoder}
\end{figure}

The decoder at the destination is a multistage decoder in which when decoding at level $l$, we treat all the subsequent levels $l'> l$ as unknown and regard the decoding at all the previous levels $\tilde{l}<l$ as correct. For this decoder, since we only consider one level at a time, the decoding performed at the $l$th level only deals with the code over $\mbb{F}_{q_l}$. Therefore, using the multilevel coding with multistage decoding for compute-and-forward only requires $\mbb{Z}$-module homomorphisms (from $\Zw$ to $\mbb{F}_{q_1}\times\ldots\times\mbb{F}_{q_L}$) instead of ring homomorphisms. This can also be seen from the examples provided in Section~\ref{sec:const} where it is evident that when considering only one level, the equivalent constellation is a modulo version of a $q_l$-ary pulse amplitude modulation (PAM). In the following, we describe the details of the proposed decoding algorithm.

At the first stage of decoding, given $b_{1}^1, b_{2}^1\in \mbb{F}_{q_1}$, in order to decode the first data stream, the decoder first computes the \textit{a posteriori} probabilities given by
\begin{align}
    &\Pp\left(\tilde{c}_{R}^1[n] = c | y[n]\right) \propto \underset{\underset{b_{1}^1c_{1}^1\oplus b_{2}^1c_{2}^1=c}{c_{1}^1,c_{2}^1\in\mbb{F}_{q_1}:}}{\sum} ~~\underset{\underset{l\in\{2,\ldots L\}}{c_{1}^l,c_{2}^l\in\mbb{F}_{q_l}}}{\sum}\nonumber \\
    &\exp\left[-\left\| h_1\mc{M}(c_{1}^1,\ldots,c_{1}^L)+h_2\mc{M}(c_{2}^1,\ldots,c_{2}^L) - y[n]\right\|^2\right],
\end{align}
for all $c\in\mbb{F}_{q_1}$ and for each codeword dimension $n$. According to these \textit{a posteriori} probabilities, the decoder forms the first level's estimate given by
\begin{equation}
    \hat{\mathbf{c}}_{R}^1 = \underset{\mathbf{c}\in C^1}{\arg\max}\prod_{n=1}^N \Pp\left(\tilde{c}_{R}^1[n] = c[n] | y[n]\right),
\end{equation}
where $c[n]$ denotes the $n$th element of the codeword $\mathbf{c}\in C^1$.

At the $l$th level, $l\in\{2,\ldots,L\}$, we assume the decoding at the levels $1,\ldots, l-1$ is correct and regards the signals from the levels $l+1,\ldots,L$ as unknown. Given $(b_{1}^l , b_{2}^l)$, the destination computes the corresponding \textit{a posteriori} probabilities given by
\begin{align}
    &\Pp\left(\tilde{c}_{R}^l[n] = c | y[n], \hat{c}_{R}^1[n],\ldots,\hat{c}_{R}^{l-1}[n]\right) \propto \nonumber \\ &\underset{\underset{\underset{\tilde{l}\in\{1,\ldots,l-1\}} {b_{1}^{\tilde{l}} c_{1}^{\tilde{l}}\oplus b_{2}^{\tilde{l}} c_{2}^{\tilde{l}} = \hat{c}_{R}^{\tilde{l}}[n]}} {c_{1}^{\tilde{l}},c_{2}^{\tilde{l}}\in\mbb{F}_{q_{\tilde{l}}}:}}{\sum} ~~
    \underset{\underset{b_{1}^2c_{1}^2\oplus b_{2}^2c_{2}^2=l}{c_{1}^2,c_{2}^2\in\mbb{F}_q:}}{\sum} ~~
    \underset{\underset{l'\in\{l+1,\ldots,L\} } {c_{1}^{l'},c_{2}^{l'}\in\mbb{F}_{q_{l'}}}}{\sum}\nonumber \\
    &\exp\left[-\left\| h_1\mc{M}(c_{1}^1,\ldots,c_{1}^L)+h_2\mc{M}(c_{2}^1,\ldots,c_{2}^L) - y[n]\right\|^2\right],
\end{align}
for all $c\in\mbb{F}_{q_l}$ and for each codeword dimension $n$. Similar to the first level, the decoder then forms the $l$th level's estimate as
\begin{equation}
    \hat{\mathbf{c}}_{R}^l = \underset{\mathbf{c}\in C^l}{\arg\max}\prod_{n=1}^N \Pp\left(\tilde{c}_{R}^l[n] = c[n] | y[n], \hat{c}_{R}^1[n],\ldots,\hat{c}_{R}^{l-1}[n]\right).
\end{equation}

\begin{remark}
    As shown in \cite{KofmanZehaviShamai}, for traditional multilevel coding scheme, adding an interleaver/deinterleaver pair prevents burst error propagation to the next level and hence improves the error probability. Similarly, for the proposed multilevel compute-and-forward scheme, one can also add an interleaver/deinterleaver pair at the encoder/decoder of each level to mitigate the effect of erroneous decoding at the previous first levels. Another way to potentially further lower the error probability is to perform iterative multistage decoding proposed in \cite{krishna00}. However, since the focus of this paper is on the analysis of achievable computation rate instead of error probability, we do not pursue these potential extensions.
\end{remark}

\subsection{Suboptimal Decoders}\label{sec:para_dec}
Motivated by the decoding algorithm in \cite{forney2000} for Construction D lattices, a suboptimal but less complex decoding algorithm can be implemented for the proposed constellations with the homomorphism given in \eqref{eqn:general_homo} as follows. Let us first assume that the channel gains are equal to 1 and $\gamma=1$ for the sake of simplicity. One observes that in the homomorphism \eqref{eqn:general_homo}, $\phi_1$ appears in the coefficient of every term except for the first term. Note that all the $\phi_l$s are relatively prime. Hence, when decoding the first codeword one can first knock out the contribution from all the other codewords by simply forming
\begin{equation}
    \mathbf{y}^1 = \mathbf{y} \mod\phi_1\Zw.
\end{equation}
The decoder then decodes to $\hat{\mathbf{c}}_R^1$ from $\mathbf{y}_1'$. For the levels $l\in\{2,\ldots,L-1\}$, the decoder subtracts all the effects from the previous levels and knocks out all the contributions from the next levels via forming
\begin{equation}
    \mathbf{y}^l = \left(\mathbf{y} - \sum_{s=1}^{l-1} \Pi_{i=1, i\neq s}^L \phi_{i} \hat{\mathbf{c}}_R^{s}\right) \mod\phi_l\Zw.
\end{equation}
The decoder then forms $\hat{\mathbf{c}}_R^l$ the output of the decoder at the $l$th level from $\mathbf{y}^l$. This makes the channel experienced by the $l$th coded stream a single level additive $\mod\phi_l\Zw$ channel. As mentioned in \cite{forney2000}, the above procedure will cause suboptimality only in the low SNR regime. In the last level of decoding, one does not have to do the modulo operation as there is only one level left. Therefore, the decoder at the last level directly decodes $\hat{\mathbf{c}}_R^L$ from
\begin{equation}
    \mathbf{y}^L = \left(\mathbf{y} - \sum_{s=1}^{L-1} \Pi_{i=1, i\neq s}^L \phi_{i} \hat{\mathbf{c}}_R^{s}\right).
\end{equation}
We summarize the procedure of the suboptimal decoder in Fig.~\ref{fig:sub_dec}. It should be noted that when channel coefficients are not $h_1=h_2=1$, similar to \cite{nazer2011CF}, one can first use an linear minimum mean squared error (MMSE) estimator to approximate the channel coefficients to a pair of Eisenstein integers and then quantize the scaled received signal to the Eisenstein integer combination of transmitted signals.
\begin{figure}
    \centering
    \includegraphics[width=4in]{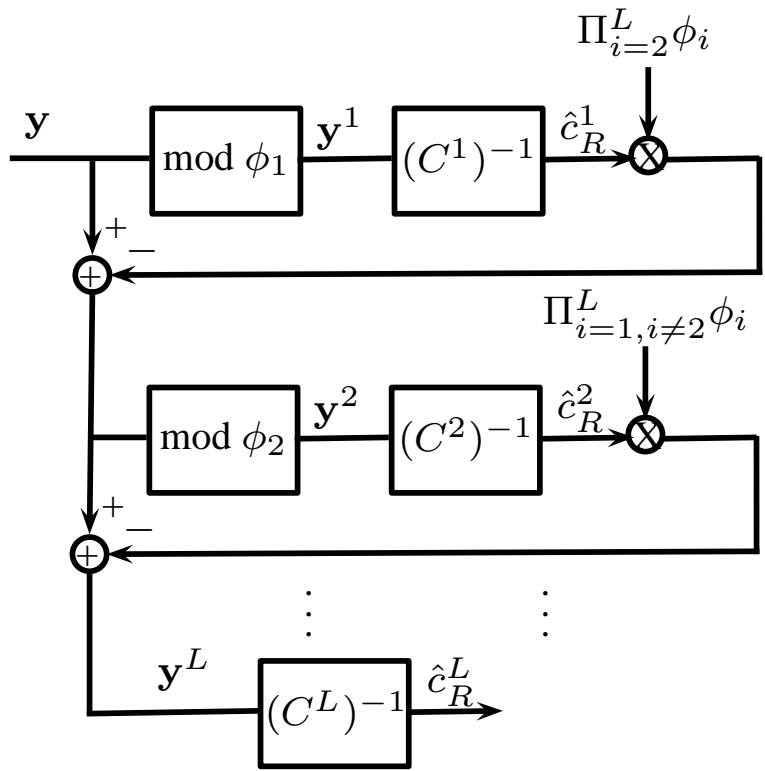}
    \caption{The proposed suboptimal decoder for the multilevel coding/multistage decoding scheme with the second proposed family of constellations and homomorphisms in \eqref{eqn:general_homo}.}
    \label{fig:sub_dec}
\end{figure}

A decoding algorithm that is even less complex and causes more rate reduction is given in Fig.~\ref{fig:para_dec} which is referred to as the parallel decoder due to that it can be implemented in a parallel fashion. This decoder simultaneously forms
\begin{equation}
    \mathbf{\tilde{y}}^l = \mathbf{y} \mod\phi_l\Zw,
\end{equation}
and then directly decodes $\hat{\mathbf{c}}_R^l$. As mentioned before, since $\phi_l$s are relatively prime, the modulo operation will get rid of the contributions from all but the $l$th level. However, in addition to having a $\mod \phi_l\Zw$ channel, this decoder also gives away the knowledge of previously decoded codeword and hence is worse than the previous one.

\begin{figure}
    \centering
    \includegraphics[width=4in]{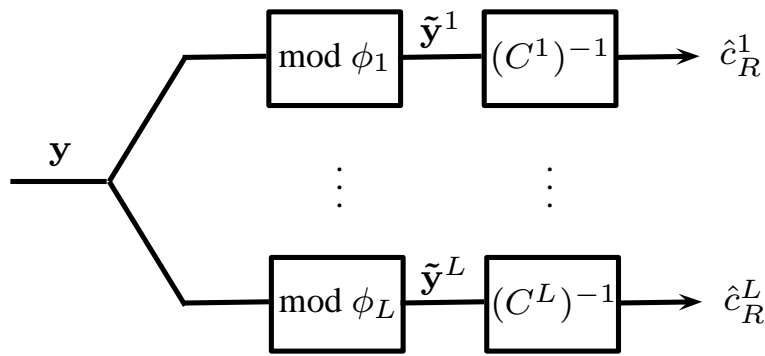}
    \caption{The proposed parallel decoder for the multilevel coding/multistage decoding scheme with the second proposed family of constellations and homomorphisms in \eqref{eqn:general_homo}.}
    \label{fig:para_dec}
\end{figure}

\subsection{Achievable Computation Rate}
In \cite{WachsmannFischerHuber99}, using the chain rule of mutual information \cite{cover91}, Wachsmann \textit{et al.} show that multilevel coding incurs no loss in terms of the achievable information rate for point to point communication. The same proof works for our construction as well, which we summarize here. Let $C^1, \ldots, C^L$ be the codebooks used for level $1,\ldots,L$, respectively, and let $\msf{C}^1,\ldots,\msf{C}^2$ be the corresponding random variables. Also, notice that the mapping between $(C^1,\ldots,C^L)$ and $\mc{A}$ is bijective. One has that
\begin{align}
    R_{\text{AWGN}} &= I(\msf{Y};\msf{A}) = I(\msf{Y};\mc{M}(\msf{C^1},\ldots,\msf{C^L}))\nonumber \\
    &\overset{(a)}{=}I(\msf{Y};\msf{C^1},\ldots,\msf{C^L}) = \sum_{l=1}^L I(\msf{Y};\msf{C^l}|\msf{C^1},\ldots,\msf{C^{l-1}}),
\end{align}
where (a) is due to the fact that $\mc{M}$ is bijective.

Now, we provide the achievable information rates of the proposed schemes for compute-and-forward. We first consider the case when multilevel coding/multistage decoding is not used. We restrict ourself to codes over fields (i.e., codes over rings are not considered in this paper). For the first proposed family in Section~\ref{sec:const1}, one can choose to directly work over $\mbb{F}_{q^2}$. Also, for the special subclass of second proposed family of constellations in Section~\ref{sec:const2_sp}, one can endow a matrix multiplication to $\mbb{F}_q^2$ so that again one can directly work over $\mbb{F}_{q^2}$. Now let $C$ be a linear codebook over $\mbb{F}_{q^2}$ and let $\msf{C}_1$ and $\msf{C}_2$ be the corresponding random variables at source nodes 1 and 2, respectively. The achievable computation rate of directly working over $\mbb{F}_{q^2}$ can be written as
\begin{equation}\label{eqn:R_GF}
    R_{\text{direct}} = \underset{b_1,b_2\in\mbb{F}_{q^2}}{\max} I(\msf{Y};b_1\msf{C_1}\oplus b_2\msf{C_2}),
\end{equation}
where the subscript "direct" stands for that we directly work over the extension field.

For the case when multilevel coding/multistage decoding scheme is adopted (which works for all the proposed constellations), let $C^1,\ldots,C^L$ be the linear codebooks adopted for level $1,\ldots,L$, respectively, and let $\msf{C}_k^1,\ldots,\msf{C}_k^L$ be the corresponding random variables at source node $k$. One has the achievable computation rate given by
\begin{align}\label{eqn:R_MLC}
    R_{\text{MLC}} &= \underset{\underset{l\in\{1,\ldots,L\}}{b_{1}^l,b_{2}^l~\in\mbb{F}_{q_l}}}{\max} I(\msf{Y};b_{1}^1\msf{C}_1^1\oplus b_2^1\msf{C}_2^1, \ldots, b_{1}^L\msf{C}_1^L\oplus b_2^L\msf{C}_2^L) \nonumber \\
    &=\underset{\underset{l\in\{1,\ldots,L\}}{b_{1}^l,b_{2}^l~\in\mbb{F}_{q_l}}}{\max} \sum_{l=1}^L I(\msf{Y};b_{1}^l\msf{C}_1^l \oplus b_2^l\msf{C}_2^l|b_{1}^1\msf{C}_1^1\oplus b_2^1\msf{C}_2^1,\ldots,b_{1}^1\msf{C}_1^{l-1}\oplus b_2^1\msf{C}_2^{l-1}) \nonumber \\
    &=\sum_{l=1}^L I(\msf{Y};b_{1}^{*l}\msf{C}_1^l \oplus b_2^{*l}\msf{C}_2^l|b_{1}^{*1}\msf{C}_1^1\oplus b_2^{*1}\msf{C}_2^1,\ldots,b_{1}^{*1}\msf{C}_1^{l-1}\oplus b_2^{*1}\msf{C}_2^{l-1}),
\end{align}
where $b_k^{*l}$ for $k\in\{1,2\}$ and $l\in\{1,\ldots,L\}$ are the maximizers. Note that this information rate may be achieved by using a \textit{good} linear code at the $l$th level with the rate set to be
\begin{equation}
    R_{\text{MLC}}^l = I(\msf{Y};b_{1}^{*l}\msf{C}_1^l \oplus b_2^{*l}\msf{C}_2^l|b_{1}^{*1}\msf{C}_1^1\oplus b_2^{*1}\msf{C}_2^1,\ldots,b_{1}^{*1}\msf{C}_1^{l-1}\oplus b_2^{*1}\msf{C}_2^{l-1}).
\end{equation}

\begin{remark}
When we consider the proposed constellations with $q^2$ elements, it should be noted that $R_\text{direct}$ and $R_\text{MLC}$ are in general not the same. It is because for $\msf{C}_k = (\tilde{\msf{C}}_k^1,\tilde{\msf{C}}_k^2)$, given $b_1,b_2\in\mbb{F}_{q^2}$, there may not exist $\tilde{b}_{1}^1,\tilde{b}_{1}^2,\tilde{b}_{2}^1,\tilde{b}_{2}^2~\in\mbb{F}_q$ such that
\begin{equation}\label{eqn:r_q_eq_r_q2}
    b_1\msf{C_1} \oplus b_2\msf{C_2} = (\tilde{b}_{1}^1\tilde{\msf{C}}_1^1\oplus \tilde{b}_2^1\tilde{\msf{C}}_2^1, \tilde{b}_{1}^2\tilde{\msf{C}}_1^2\oplus \tilde{b}_2^2\tilde{\msf{C}}_2^2),
\end{equation}
and $\tilde{\msf{C}}_k^1, \tilde{\msf{C}}_k^2$ are valid codewords over $\mbb{F}_q$. In what follows, we include the idea of flexible decoding \cite{Brett11} so that the multilevel coding/multistage decoding scheme can recover all the combinations in $\mbb{F}_{q^2}$ and can potentially do more.
\end{remark}

One can also get the achievable rates with the suboptimal decoder and that with the parallel decoder as follows.
\begin{align}\label{eqn:R_sub}
    R_{\text{sub}} &=\underset{\underset{l\in\{1,\ldots,L\}}{b_{1}^l,b_{2}^l~\in\mbb{F}_{q_l}}}{\max} \sum_{l=1}^L I(\msf{Y}^l;b_{1}^l\msf{C}_1^l \oplus b_2^l\msf{C}_2^l) \nonumber \\
    &=\sum_{l=1}^L I(\msf{Y}^l;b_{1}^{*l}\msf{C}_1^l \oplus b_2^{*l}\msf{C}_2^l),
\end{align}
and
\begin{align}\label{eqn:R_sub}
    R_{\text{para}} &=\underset{\underset{l\in\{1,\ldots,L\}}{b_{1}^l,b_{2}^l~\in\mbb{F}_{q_l}}}{\max} \sum_{l=1}^L I(\msf{\tilde{Y}}^l;b_{1}^l\msf{C}_1^l \oplus b_2^l\msf{C}_2^l) \nonumber \\
    &=\sum_{l=1}^L I(\msf{\tilde{Y}}^l;b_{1}^{*l}\msf{C}_1^l \oplus b_2^{*l}\msf{C}_2^l),
\end{align}
where again $b_k^{*l}$ for $k\in\{1,2\}$ and $l\in\{1,\ldots,L\}$ are the maximizers. Again, these rates may be achieved by setting the rate at $l$th level to be $R_{\text{sub}}^l = I(\msf{Y}^l;b_{1}^{*l}\msf{C}_1^l \oplus b_2^{*l}\msf{C}_2^l)$ and $R_{\text{para}}^l = I(\msf{\tilde{Y}}^l;b_{1}^{*l}\msf{C}_1^l \oplus b_2^{*l}\msf{C}_2^l)$, respectively

\subsection{Flexible Decoding}\label{sec:flex_dec}
For the proposed constellations with $q^2$ elements, i.e., those in the first proposed family or in the special class of the second proposed family, there is an interesting extension that may potentially increase the achievable computation rates. The idea is to restrict the codes used in the two levels (they are over the same field $\mbb{F}_q$) to be the same, i.e., $C^1 = C^2$. By doing this, one can incorporate the idea of flexible decoding \cite{Brett11} into our framework. By doing this, in addition to the original choice we have had
\begin{equation}
    b_1^1 \mathbf{c}_1^1\oplus b_2^1 \mathbf{c}_2^1 \text{~and~} b_1^2 \mathbf{c}_1^2\oplus b_2^2 \mathbf{c}_2^2,
\end{equation}
where $b_1^1, b_1^2, b_2^1, b_2^2\in\mbb{F}_q$, one can decode to something else. For example,
\begin{equation}
    \tilde{b}_1^1 \mathbf{c}_1^1\oplus \tilde{b}_2^2 \mathbf{c}_2^2 \text{~and~} \tilde{b}_1^2 \mathbf{c}_1^2\oplus \tilde{b}_2^1 \mathbf{c}_2^1,
\end{equation}
where $\tilde{b}_1^1, \tilde{b}_1^2, \tilde{b}_2^1, \tilde{b}_2^2\in\mbb{F}_q$. More precisely, one can decode the received signal to
\begin{equation}\label{eqn:flexble_decoding}
    [\mathbf{c}_R^1,\mathbf{c}_R^2]^T = [\mathbf{B}_1 \mathbf{B}_2]
    \left[
                                 \begin{array}{c}
                                   \mathbf{c}_1^1 \\
                                   \mathbf{c}_1^2 \\
                                   \mathbf{c}_2^1 \\
                                   \mathbf{c}_2^2 \\
                                 \end{array}
                               \right],
\end{equation}
where $\mathbf{B}_1$ and $\mathbf{B}_2$ are chosen from $\mc{B}$ the set of all 2 by 2 full-rank matrices with elements in $\mbb{F}_q$. This is because now all the codes used are identical so that $\mathbf{c}_1^1, \mathbf{c}_1^2, \mathbf{c}_2^1, \mathbf{c}_2^2$ are codewords from a same linear code. This approach allows rich choices of functions that one can decode to and hence may result in a higher rate in general. We now summarize the computation rates achieved by the proposed scheme with flexible decoding in the next theorem.

\begin{theorem}\label{thm:flexible}
    The achievable computation rate for the proposed constellations with the proposed multilevel coding scheme and with flexible decoding is given by
    \begin{align}
        R_{\text{flex}} \leq \underset{\mathbf{B}_1,\mathbf{B}_2\in\mc{B}}{\max} \min &\left\{ I(\msf{Y};\msf{C}_R^1|\msf{C}_R^2) , I(\msf{Y};\msf{C}_R^2|\msf{C}_R^1) \vphantom{\frac{1}{2}}, \right.\nonumber \\
        &  \left.\frac{1}{2}I(\msf{Y};\msf{C}_R^1,\msf{C}_R^2), I(\msf{Y};\msf{C}_R^1,\msf{C}_R^2|\msf{C}_R^1\oplus\msf{C}_R^2) \right\},
    \end{align}
    where $\msf{C}_R^1$ and $\msf{C}_R^2$ are given in \eqref{eqn:flexble_decoding} and are dependant on the choices of $\mathbf{B}_1$ and $\mathbf{B}_2$.
\end{theorem}
\begin{IEEEproof}
    This theorem is a $\mbb{F}_q$ version of Theorem 1 in \cite{Brett11} and the proof is hence omitted.
\end{IEEEproof}

Note that multiplication over $\mbb{F}_{q^2}$ can be represented as multiplication of a matrix and a vector over $\mbb{F}_q$. Thus, setting $C^1=C^2$ enables the proposed scheme with flexible decoding to recover all the linear combinations of codewords over $\mbb{F}_{q^2}$. For example, let $b\in \mbb{F}_{25}$ whose multiplication is defined by the irreducible polynomial $x^2+2x+4$ as in Example~\ref{exp:ring_Eis} and $b=b^1x+b^2$ with $b^1, b^2\in \mbb{F}_5$. Also, let $\mathbf{c} = \mathbf{c}^1 x + \mathbf{c}^2$ where $\mathbf{c}^1\in C^1$ and $\mathbf{c}^2\in C^2$ over $\mbb{F}_5$. Then, one has that
\begin{align}
    b\cdot \mathbf{c} &= b^1 \mathbf{c}^1 x^2 + (b^2 \mathbf{c}^1 \oplus b^1 \mathbf{c}^2) x + b^2 \mathbf{c}^2 \nonumber \\
    &= (b^2 \mathbf{c}^1 \oplus 3b^1 \mathbf{c}^1 \oplus b^1 \mathbf{c}^2) x + (b^2 \mathbf{c}^2 \oplus b^1 \mathbf{c}^1) \nonumber \\
    &=\left[
        \begin{array}{cc}
          b^2 \oplus 3b^1 & b^1 \\
          b^1 & b^2 \\
        \end{array}
      \right]
      \left[
        \begin{array}{c}
          \mathbf{c}^1 \\
          \mathbf{c}^2 \\
        \end{array}
      \right].
\end{align}
Therefore, every linear combination of codewords over $\mbb{F}_{25}$ of the form $b_1\mathbf{c}_1 \oplus b_2\mathbf{c}_2$ can be represented as a linear combination of codewords over $\mbb{F}_5$ by choosing
\begin{equation}\label{eqn:B1}
    \mathbf{B}_1 = \left[
        \begin{array}{cc}
          b_1^2 \oplus 3b_1^1 & b_1^1 \\
          b_1^1 & b_1^2 \\
        \end{array}
      \right],
\end{equation}
and
\begin{equation}\label{eqn:B2}
    \mathbf{B}_2 = \left[
        \begin{array}{cc}
          b_2^2 \oplus 3b_2^1 & b_2^1 \\
          b_2^1 & b_2^2 \\
        \end{array}
      \right].
\end{equation}

However, this still does not mean that the proposed multilevel coding scheme together with flexible decoding would achieve the same rate with that provided by the code over $\mbb{F}_{q^2}$. The reasons for this are twofold. One is that setting the linear codes to be the same imposes an extra constraint on the rate as shown in the last term of Theorem~\ref{thm:flexible}. Second, the symmetric capacity may not touch the boundary of the sum rate limit for the underlying MAC channel. It is interesting to see when setting the codes to be the same would not result in the above penalty. Currently we have been able to identify some special cases for which one can ignore the second penalty. For example, we have the following theorem which includes the symmetric bidirectional relaying problem studied in \cite{wilson10} and the simulation setup in Fig.~\ref{fig:MLC_eq} as special cases.
\begin{theorem}\label{thm:sym_cap}
    Let $h_1=|h_1|e^{j\theta}$ and $h_2=|h_2|e^{j\theta}$, i.e., they have a same phase $\theta$. Also let the functions for the two levels to be the same, i.e., $b_k^1=b_k^2=b_k\in\mbb{F}_q$. Then the symmetric capacity always lies on the boundary of the sum rate limit of the underlying MAC channel.
\end{theorem}
\begin{proof}
    See Appendix~\ref{apx:sym_rate}.
\end{proof}
It is worth mentioning that despite the above extra penalties, for many cases the proposed multilevel coding scheme with flexible decoding may in fact result in a higher achievable computation rate than that provided by directly coding over $\mbb{F}_{q^2}$. This can be seen from the fact that, in general, there are many full-rank matrices $\mathbf{B}_1$ and $\mathbf{B}_2$ which are not in the form of \eqref{eqn:B1} and \eqref{eqn:B2}, respectively. i.e., there exist many decoding functions for the proposed scheme with flexible decoding which can not be provided by the scheme directly coding over $\mbb{F}_{q^2}$. This makes the proposed scheme to be robust to phase shift. Similar results can be found in \cite{Brett11}.

\section{Simulation Results}\label{sec:simulation}
In this section, we use the Monte-Carlo method to evaluate the achievable computation rates. We first compare the computation rates achieved by using different mappings for the proposed constellations. After this, we provide comparisons on the performance of the proposed constellations with the proposed multilevel coding/multistage decoding and that with direct coding over $\mbb{F}_{q^2}$. In order to show that the proposed scheme indeed can approach those theoretic limits with reasonable complexity, we also simulate the proposed scheme with an ensemble of linear codes with iterative decoding. Recently, it has been shown in \cite{Kudekar11} \cite{Kudekar12} \cite{Pfister13} that the ensemble of spatially-coupled LDPC codes (or LDPC convolutional codes) \cite{Sridharan04} \cite{Lentmaier05} \cite{Lentmaier10} universally achieve the capacity for the class of binary memoryless symmetric channels under belief propagation (BP) decoding \cite{urbanke_book}. Motivated by this success, we choose the linear code at the $l$th level to be an non-binary spatially-coupled LDPC code over $\mbb{F}_{q_l}$ with BP decoding at the receiver. The ensemble we will use is the $(d_l,d_r,L')$ ensemble introduced in \cite[Section II-A]{Kudekar11} (here we use $L'$ to denote the coupling length instead of $L$ to avoid confusion with the number of levels $L$). This can be regarded as the extension of the LDA lattices in \cite{diPietro12} \cite{diPietro13} or the SCLDA lattices in \cite{Engin13ITW} to the proposed product construction.

\subsection{Comparison of Different Mappings}
For point-to-point communication, one of the most popular and most frequently used labeling strategy for multilevel coding is the one introduced by Ungerboeck \cite{Ungerboeck82}. The design guideline of this labeling strategy is to maximize the minimum intra-subset Euclidean distance. However, as mentioned in Section~\ref{sec:const2_sp}, this labeling strategy does not guarantee the homomorphism property that has been shown crucial for compute-and-forward problem. In what follows, we present some comparisons of the proposed mappings. We will show that indeed, as suggested in \cite{Feng10} and \cite{Engin12}, homomorphisms are crucial for compute-and-forward.

We consider the constellation with the Eisenstein prime $\phi=3+2\omega$. Note that $|\phi|^2=7$ is congruent to $1\mod 3$; therefore, this belongs to the special subclass of the second proposed family of constellations with $|\mc{A}|=49$. As mentioned before, for constellations in this subclass, one can freely choose any two linear independent vectors to generate a $\mbb{Z}$-module homomorphism. In the following, we briefly compare the achievable computation rates for two such homomorphisms, namely the one in \eqref{eqn:module_homo} (which is referred to as mapping 1) and the one uses the idea of Ungerboeck in \eqref{eqn:module_homo_ungerboeck}, also in Fig.~\ref{fig:ungerboeck2} (which is referred to as mapping 2). In Fig.~\ref{fig:MLC_49pt_homo}, we simulate the achievable computation rates for the case when $h_1 = h_2 = 1$ in order to avoid unnecessary distraction from the self-interference \cite{nazer2011CF}. One observes that the sum rates provided by the two mappings are the same. In the following simulations, since we will be focusing on the sum rates, we will only consider the proposed mapping 1 as the two mappings provide the same sum rates. On the other hand, the rates achieved at each level for the two mappings are quite different. This is because the minimum distances and the numbers of nearest neighbors at each level of the two mappings are different. Therefore, although the sum rates are the same, in practice, one can choose the mapping that is more suitable to the problem at hand. The cut-set upper bound $\log(1+\text{SNR})$ and the computation rates achieved by infinitely-dimensional lattice $\log(.5+\text{SNR})$ \cite{wilson10} are also plotted for comparison. One observes that there is a gap between the theoretic bounds and the proposed scheme in the moderate SNR regime. This is the shaping loss suffered by the separation-based compute-and-forward.

\begin{figure}
    \centering
    \includegraphics[width=5in]{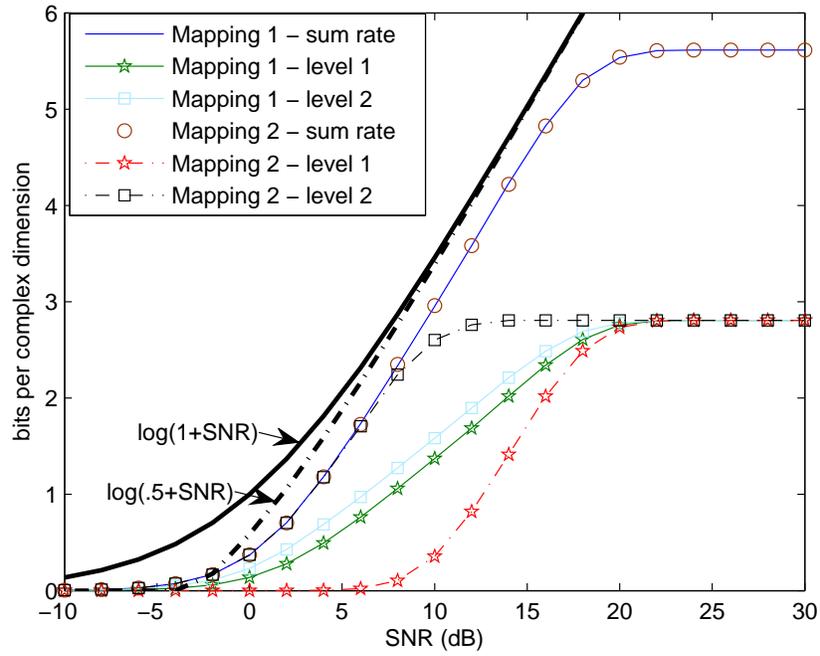}
    \caption{Achievable rates of the proposed scheme with the mapping in \eqref{eqn:module_homo} and that with the mapping in \eqref{eqn:module_homo_ungerboeck}.}
    \label{fig:MLC_49pt_homo}
\end{figure}

We then compare the performance of the proposed labeling and the labeling obtained from applying Ungerboeck's idea blindly, which we refer to as naive Ungerboeck labeling. Again, we consider the constellation with the Eisenstein prime $\phi=3+2\omega$. For the proposed multilevel coding scheme, we use the homomorphism in \eqref{eqn:module_homo} (mapping 1) as a labeling strategy. Moreover, for this family of constellations, a mapping obtained from the idea of Ungerboeck is given in Fig.~\ref{fig:ungerboeck} where only the 49 points inside the big hexagon are used. The channel coefficients are again set to be $h_1 = h_2 = 1$. One observes in Fig.~\ref{fig:MLC_49pt_labeling} that the proposed labeling substantially outperforms the naive Ungerboeck labeling in the high SNR regime. This is because Ungerboeck's labeling does not guarantee the homomorphism which has been shown to be crucial for compute-and-forward. For example, $1$ and $1+\omega$ in $\mbb{C}$ correspond to $[5,0]$ and $[1,0]$ in $\mbb{F}_7$, respectively, and $1+(1+\omega) = 2 +\omega$ in $\mbb{C}$ corresponds to $[6,1]$ in $\mbb{F}_7^2$ which is not equal to $[3,0]+[1,0]=[4,0]$ in $\mbb{F}_7^2$. The lack of a homomorphism renders the Ungerboeck's labeling ineffective for compute-and-forward in terms of the achievable computation rate. This coincides with and reinforces the main observation in \cite{Feng10} and \cite{Engin12} that homomorphisms are crucial for compute-and-forward.

\begin{figure}
    \centering
    \includegraphics[width=5in]{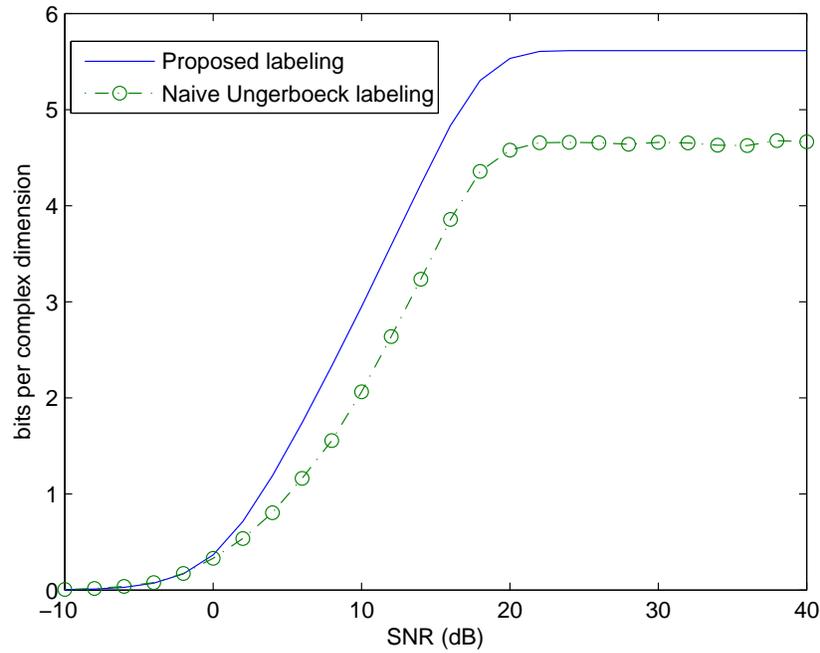}
    \caption{Achievable rates of the proposed constellation with the proposed labeling and that with the naive Ungerboeck labeling.}
    \label{fig:MLC_49pt_labeling}
\end{figure}

\subsection{MLC vs. Coding over $\mbb{F}_{q^2}$}
Here, we compare the achievable computation rates of the proposed constellations with the proposed MLC and the one directly coding over $\mbb{F}_{q^2}$. The Eisenstein prime that we use in the following simulations is $\phi = 5$, i.e., $|\mc{A}| = 25$. Since $5$ is congruent to $2\mod 3$, this belongs to the first proposed family of constellations. Hence, one can choose either to directly carry out the separation-based scheme with the ring homomorphism given in Example~\ref{exp:ring_Eis} or to implement the proposed multilevel coding and multistage decoding scheme. In what follows, we do both and compare the achievable rates of these approaches given in \eqref{eqn:R_GF} and \eqref{eqn:R_MLC}.

In Fig.~\ref{fig:MLC_eq}, we show the achievable rates of the proposed constellation with multilevel coding where each level employs a linear code over $\mbb{F}_5$ and that with a linear code over $\mbb{F}_{25}$. The ring homomorphism adopted is as shown in Example~\ref{exp:ring_Eis}. For the multilevel coding scheme, the rate achieved by each level is also shown. The transmitted SNR is ranging from -10 dB to 40 dB. The channel coefficients are set to be $h_1 = h_2 = 1$ in order to simulate the scenario when there is no self-interference. In this case, one can see from the figure that using the proposed constellations with multilevel coding incur no rate loss compared to the scheme directly working over $\mbb{F}_{25}$. It is because for this case, both schemes would choose to decode the received signal to the sum of the messages (over the corresponding fields) and element-wise addition in the base field is equivalent to addition in the extension field.

\begin{figure}
    \centering
    \includegraphics[width=5in]{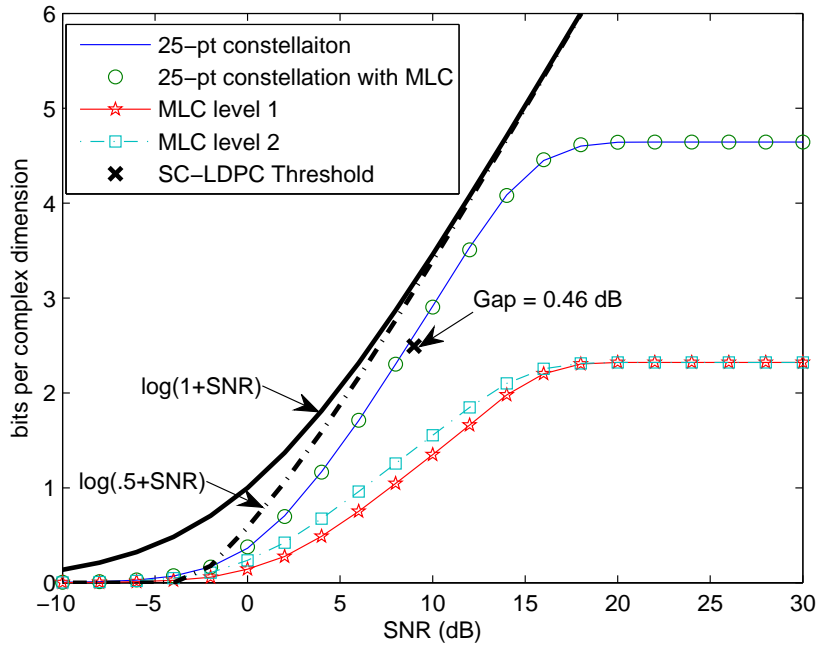}
    \caption{Achievable rates of the proposed construction with and without multilevel coding. For the one with multilevel coding, the achievable rate achieved by each level is also plotted. The channel coefficients are set to be $h_1 = h_2 = 1$.}
    \label{fig:MLC_eq}
\end{figure}

We then compare the average achievable rates of the proposed constellation in Fig.~\ref{fig:ring_homo_Eis} with and without multilevel coding. For the proposed multilevel coding scheme, we also plot the rates achieved at each level. We average over 100 pairs of channel coefficients drawn from $\mc{CN}(0,1)$ (i.e., its norm has a Rayleigh distribution). The results are shown in Fig.~\ref{fig:MLC_avg} where one can see that the scheme directly working over $\mbb{F}_{25}$ provides a slightly higher rate than that provided by the multilevel coding which works over $\mbb{F}_5$. However, the gap becomes smaller and smaller as the SNR increases. One also observes that after roughly $26$ dB, the gap becomes negligible and the proposed multilevel coding scheme over $\mbb{F}_5$ outperforms the scheme working over $\mbb{F}_{19}$. This shows that using the proposed scheme over $\mbb{F}_5$, one can perform very close to the scheme over $\mbb{F}_{25}$ and outperform the scheme over $\mbb{F}_7$ and $\mbb{F}_{19}$ in the high SNR regime with a substantially lower computational complexity.

In this figure, we also show the result of using the spatially-coupled LDPC codes with BP decoding. For the first level, the designed rate is set to be $1/2\log(5)$ so that one can directly use the $(3,6,64)$ spatially-coupled LDPC ensemble \cite{Kudekar11}. The number of variable nodes at each position, i.e., the protograph lifting factor, is chosen to be 10000; hence, the overall code length is $1.29\cdot 10^6$. For the second level, the same ensemble is used but is punctuated such that the rate becomes the one corresponding to the theoretic limit. The threshold is determined by the maximum noise variance for which no codeword errors were observed in the simulation of 10 consecutive codewords. In this figure, one can see that including the rate loss from the termination, we have been able to observe a threshold when SNR is equal to 9 dB, which is 0.46 dB away from the theoretic limit.

\begin{figure}
    \centering
    \includegraphics[width=5in]{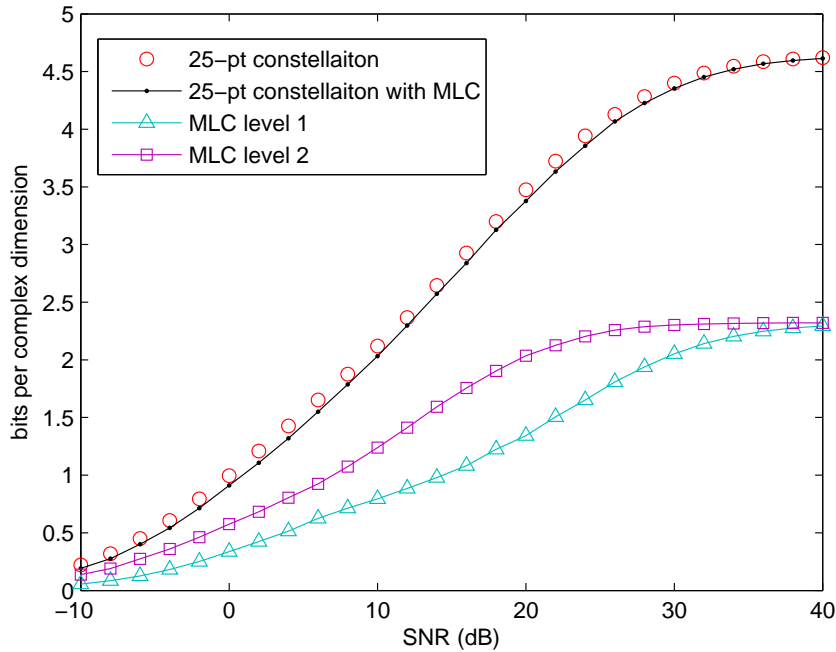}
    \caption{Average information rates of the proposed constellation with multilevel coding over $\mbb{F}_5$ and that with a linear code over $\mbb{F}_{25}$. Averaging 100 pairs of channel coefficients drawn from $\mc{CN}(0,1)$.}
    \label{fig:MLC_avg}
\end{figure}

\begin{remark}
    Although, in Section~\ref{sec:flex_dec}, the proposed scheme with flexible decoding has been suggested to potentially achieve higher rates and to recover the rates achieved by directly coding over $\mbb{F}_{q^2}$, we do not use it in the preceding simulations. It is mainly because the number of functions one can decode to grows very rapidly with $q$. In fact, the number of choices of each $\mathbf{B}_i$ is $(q^2-1)(q^2-q)$ and it is extremely time-consuming to run the Monte-Carlo simulation for all the possible choices of $\mathbf{B}_i$ and find the one that maximizes the achievable rates. Thus, efficient algorithms are called for. So far there have been some work in the literature on efficient algorithms of finding approximately optimal $\textit{linear}$ functions \cite{Wei13}. It is interesting to design efficient algorithms for the proposed scheme with flexible decoding in which the functions may \textit{not} be linear. We leave this problem to future work. Nevertheless, the results in this section suggest that even without the flexible decoding, the proposed multilevel coding scheme still performs very close to the one directly coding over $\mbb{F}_{q^2}$.
\end{remark}

\subsection{Achievable Rates for Constellations with Different Size}
In Fig.~\ref{fig:CF_7_13_19pt}, we plot the achievable computation rates for the separation-based compute-and-forward with constellations with $7$ elements, that with $13$ elements, and that with $19$ elements from \cite{Engin12}. Also, the achievable rates for the proposed multilevel coding/multistage decoding with proposed constellations with $21$ elements, that with $25$ elements, and that with $49$ elements are plotted. The channel coefficients are set to be $h_1=h_2=1$ in this figure. One can see that the proposed scheme together with the proposed constellations provide a way to extend the separation-based compute-and-forward to the high rate regime with a relatively lower complexity. Specifically, the decoding complexity for the constellations with $21$, $25$, and $49$ elements is dominated by the decoding complexity for codes over $\mbb{F}_7$, $\mbb{F}_5$, and $\mbb{F}_7$, respectively. This is at least as low (in terms of order) as using separation-based compute-and-forward for constellation with $7$ elements. Similar results can also be observed in Fig.~\ref{fig:MLC_21_25_avg} for the average (over 100 realizations) achievable rates when the channel coefficients are drawn from $\mc{CN}(0,1)$.

\begin{figure}
    \centering
    \includegraphics[width=5in]{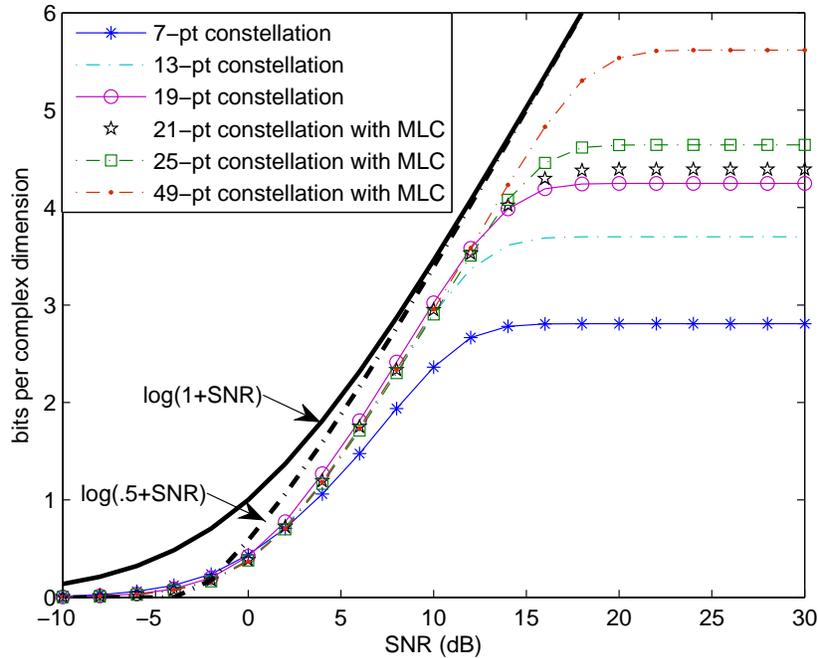}
    \caption{Achievable rates for constellations with different size.}
    \label{fig:CF_7_13_19pt}
\end{figure}

\begin{figure}
    \centering
    \includegraphics[width=5in]{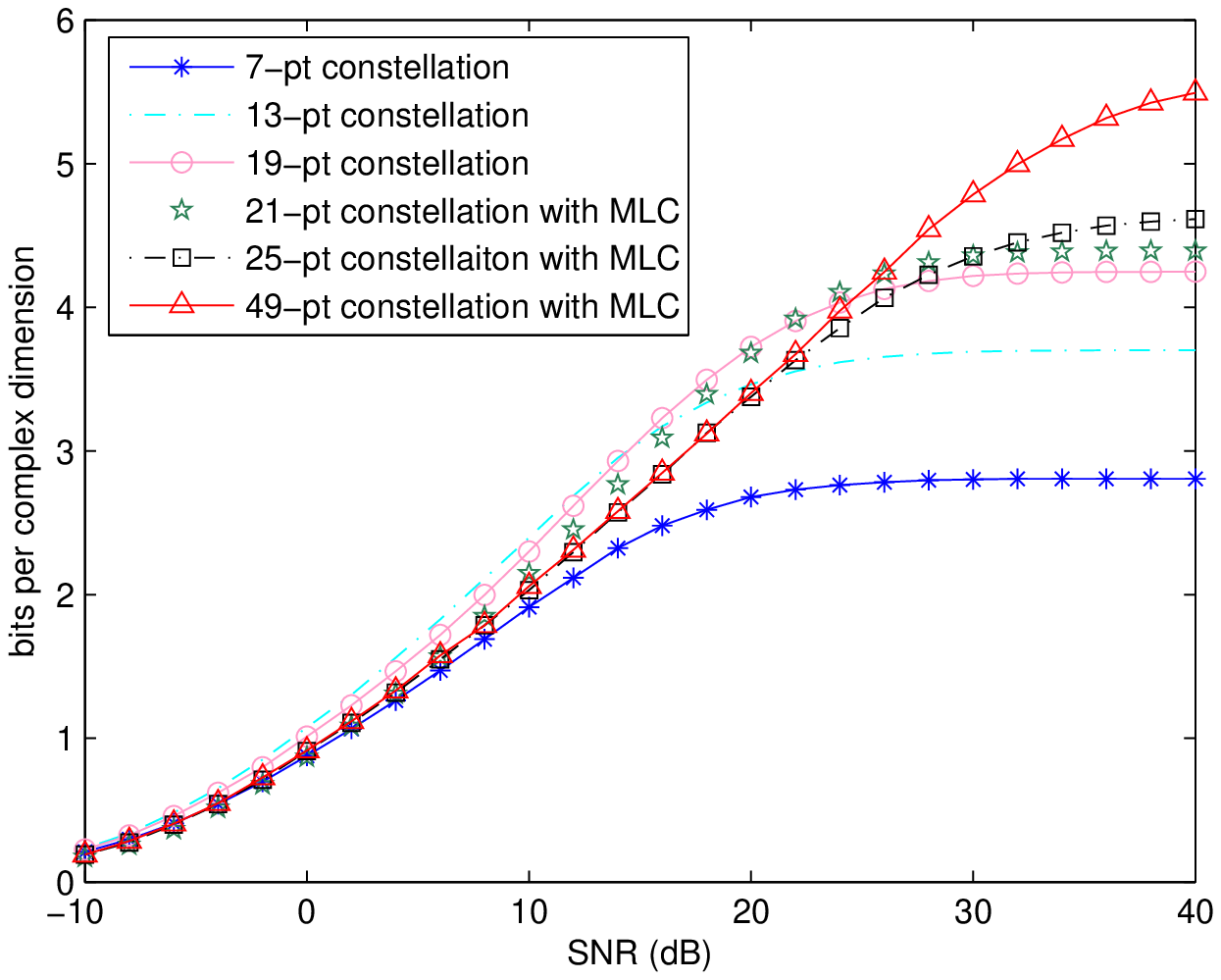}
    \caption{Average achievable rates for constellations with different size.}
    \label{fig:MLC_21_25_avg}
\end{figure}

\subsection{Comparison of Different Decoders}
In Fig.~\ref{fig:MLC_49pt_dec}, we compare the achievable computation rates for the proposed scheme with a multistage decoder, that with the suboptimal decoder, and that with the parallel decoder discussed in Section~\ref{sec:encode_decode}. The constellation adopted is the 49 elements constellation generated by $\phi_1=3+2\omega$ and its complex conjugate and the channel coefficients are set to be $h_1=h_2=1$. One observes that, as expected, the multistage decoder performs the best among these decoders as it is also the most complex one. On the other hand, although being suboptimal, the suboptimal decoder can provide rates close to that provided by the multistage decoder. For the parallel decoder, the achievable rates are much worse than that for the other two in the low SNR regime but it is still interesting in the medium and high SNR regime due to its low complexity and low latency.

\begin{figure}
    \centering
    \includegraphics[width=5in]{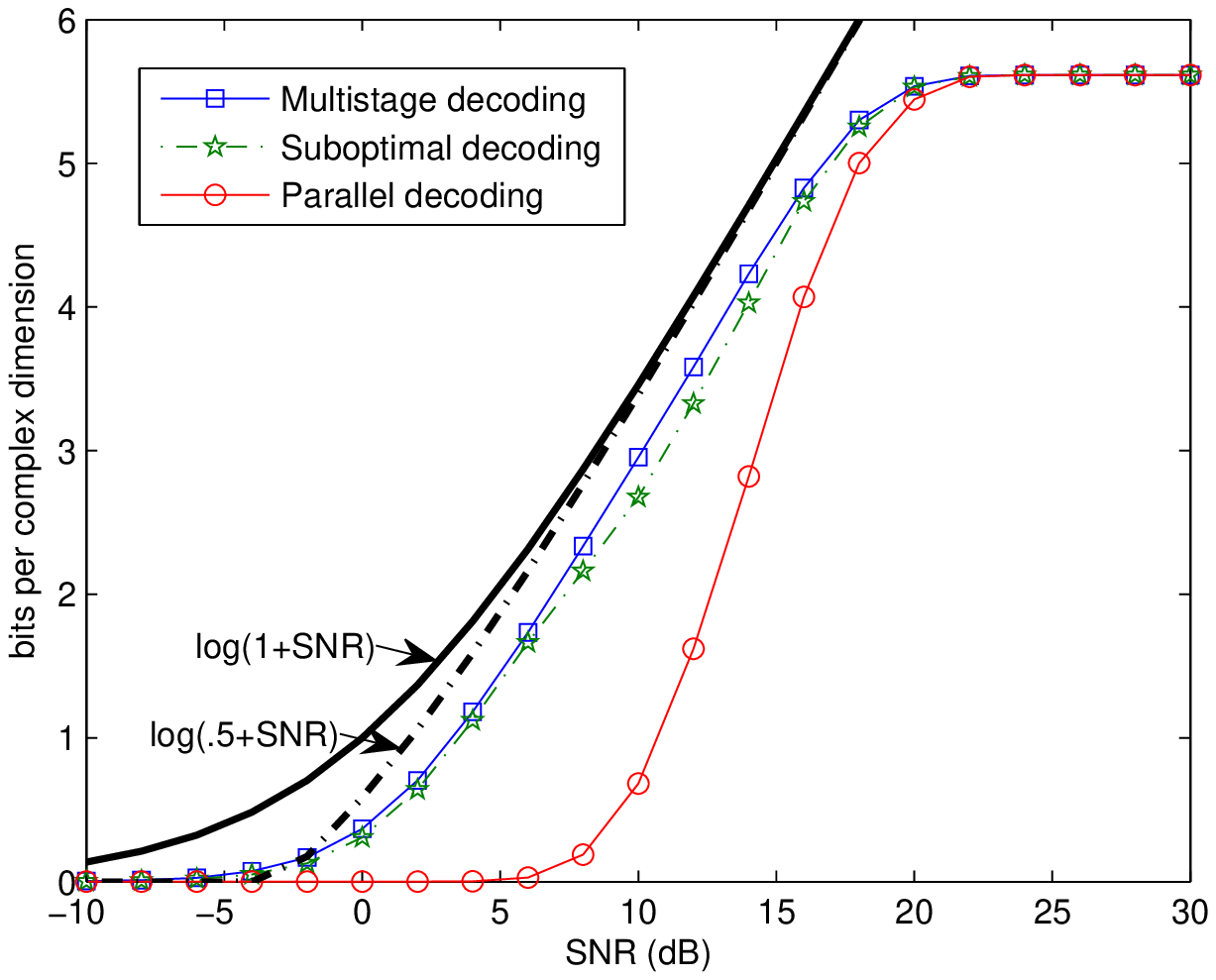}
    \caption{Average achievable rates for constellations with different size.}
    \label{fig:MLC_49pt_dec}
\end{figure}

\subsection{Comparison of Complexity}
We present a numerical result for providing a rough comparison of the decoding complexity between the proposed scheme and the one directly coding over the prime field for separation-based compute-and-forward. This can also be deemed as the comparison of the decoding complexity between the proposed product construction lattices and Construction A lattices both with the underlying codes being LDPC codes. Recall that for Construction A lattices the decoding complexity is dominated by $|\Lambda^*|$ while for the proposed product construction lattices, it only depends on the greatest divisor of $|\Lambda^*|$. For coding over $\mbb{F}_q$, we assume that a $q$-ary LDPC code is implemented and the decoding algorithm in \cite{DaveyMackay} with complexity $O(q\log(q))$ is adopted. In Fig.~\ref{fig:complexity}, we provide a comparison of the decoding complexity for Construction A lattices and the proposed product construction lattices. Note that for the proposed construction of lattices, we exclude those lattices generated by prime numbers since for those the complexity is the same as using Construction A. One observes that the proposed product construction significantly reduces the decoding complexity. Moreover, one can expect the gain to be larger and larger as the constellation size increases.

\begin{figure}
    \centering
    \includegraphics[width=5in]{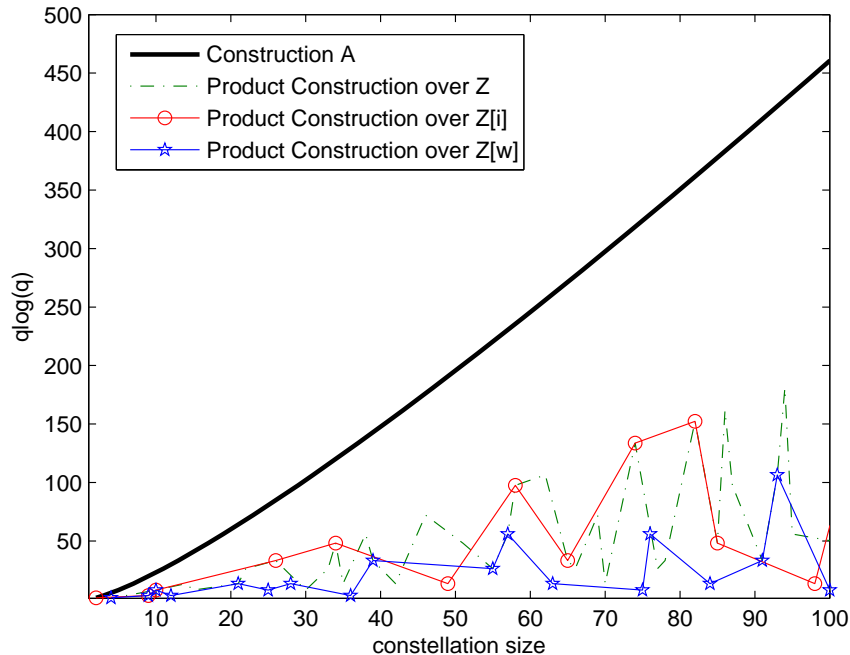}
    \caption{A comparison of the decoding complexity.}
    \label{fig:complexity}
\end{figure}

\section{Conclusions}\label{sec:conclude}
A novel construction of lattices called the product construction has been proposed. The existence of a sequence of lattices generated by the product construction that are good for MSE quantization and Poltyrev-good under multistage decoding has been shown. This has allowed us to perform lattice-based multistage compute-and-forward for achieving the same information-theoretic results in \cite{nazer2011CF} using multistage decoding.

We have used the proposed lattices to generate signal constellations that are suitable for separation-based compute-and-forward. Using the idea of multilevel coding and multistage decoding, we have proposed a low complexity scheme called separation-based compute-and-forward that would substantially reduce the decoding complexity in the high rate regime. We have also showed that the use of multilevel coding and multistage decoding incurs no essential rate loss in the regions that one would operate on. Moreover, for some special cases, the proposed scheme can be further extended by incorporating the idea of flexible decoding in \cite{Brett11} for potentially increasing the achievable computation rate.

\appendices
\section{Proof of Theorem~\ref{thm:lattice}}\label{apx:lattice}
In this appendix, we prove Theorem~\ref{thm:lattice}. Particularly, we focus on $\mbb{Z}$ the ring of integers and the proof for $\Zw$ can be obtained in a similar fashion. In this appendix, we will slightly abuse the notation and write $q \defeq \Pi_{l=1}^L p_l$.

\subsection{$\Lambda$ is a Lattice}
We first prove that $\Lambda$ is a lattice. i.e., $\Lambda$ is a discrete subgroup of $\mbb{R}^N$ which is closed under reflection and real addition. Since $\mc{M}(0,\ldots,0)=0$ and $\mathbf{0}\in C^l$, we have $\mathbf{0}\in\Lambda$. Let $\mathbf{c}^l_1,\mathbf{c}^l_2 \in C^l$ and let
\begin{align}
    \mathbf{\lambda}_1 &= \mathcal{M}(\mathbf{c}^1_1,\ldots,\mathbf{c}^L_1) + q \mathbf{\zeta}_1, \\
    \mathbf{\lambda}_2 &= \mathcal{M}(\mathbf{c}^1_2,\ldots,\mathbf{c}^L_2) + q \mathbf{\zeta}_2,
\end{align}
where $\mathbf{\zeta}_1,\mathbf{\zeta}_2\in \mbb{Z}^N$. It is clear that $\mathbf{\lambda}_1$ and $\mathbf{\lambda}_2$ are elements in $\Lambda$. One has that
\begin{align}
    \mathbf{\lambda}_1 + \mathbf{\lambda}_2 &= \mathcal{M}(\mathbf{c}^1_1,\ldots,\mathbf{c}^L_1)+ \mathcal{M}(\mathbf{c}^1_2,\ldots,\mathbf{c}^L_2) + q \mathbf{\zeta}_3 \nonumber \\
    &\overset{(a)}{=} \mathcal{M}(\mathbf{c}^1_1\oplus\mathbf{c}^1_2,\ldots,\mathbf{c}^L_1\oplus\mathbf{c}^L_2)+ q \mathbf{\zeta}'_3 \nonumber \\
    &\overset{(b)}{=} \mathcal{M}(\mathbf{c}^1_3,\ldots,\mathbf{c}^L_3) + q \mathbf{\zeta}'_3,
\end{align}
where $\mathbf{\zeta}_3, \mathbf{\zeta}'_3 \in \mbb{Z}^N$ and $\mathbf{c}^l_3\triangleq \mathbf{c}^l_1\oplus\mathbf{c}^l_2 \in C^l$. Note that (a) follows from the fact that $\mathcal{M}$ is an isomorphism, and (b) is due to the fact that $C^l$ for $l\in\{1,\ldots,L\}$ are linear codes. Now, one can see that $\mathcal{M}(\mathbf{c}^1_3,\ldots,\mathbf{c}^L_3) + q \mathbf{\zeta}_3$ is indeed an element in $\Lambda$. Moreover, choosing $\mathbf{c}_2^l$ such that $\mathbf{c}_3^l=\mathbf{0}$ for all $l\in\{1,\ldots,L\}$ and choosing $\zeta_2$ such that $\zeta'_3=\mathbf{0}$ make $\lambda_2$ the additive inverse of $\lambda_1$. Therefore, $\Lambda$ is a lattice.

\subsection{Existence of Poltyrev-Good Lattices}
We begin by noting that any lattice $\Lambda$ generated by Construction A can be written as (up to scaling) $\Lambda = \Lambda^* + p\mbb{Z}^N$, where $\Lambda^*$ is a coded level resulting from mapping a $(N,k)$ linear code to $\mbb{Z}_p^N$ via a ring isomorphism and $p\mbb{Z}^n\defeq \Lambda'$ can be viewed as an uncoded level. As shown in \cite{forney2000}, one can first reduce the received signal by performing $\mod\Lambda'$. This will make the equivalent channel a $\Lambda/\Lambda'$ channel. When the underlying linear code is capacity-achieving for the $\Lambda/\Lambda'$ channel, the probability of error for the first level can be made arbitrarily small. Moreover, by choosing $p$ arbitrarily large, the probability that one would decode to a wrong lattice point inside the same coset can be made arbitrarily small. i.e., the probability of error for the second level can be made arbitrarily small. Forney \textit{et al.} in \cite{forney2000} showed the existence of a sequence of Poltyrev-good lattices under the above two conditions.

In the following, we closely follow the steps in \cite{forney2000} to show the existence of Poltyrev-good lattices generated by our product construction. Let $p_1, p_2,\ldots, p_L$ be a collection of distinct odd primes. Similar to lattices from Construction A, one can view the lattices from product construction as $\Lambda = \Lambda^* + \Pi_{l=1}^L p_l \mbb{Z}^N$ where $\Lambda^*$ is obtained from the steps 1) and 2) in Section~\ref{sec:prod_const} and $\Pi_{l=1}^L p_l \mbb{Z}^N \defeq \Lambda'$ is an uncoded level. Similar to \cite{forney2000}, the probability of error in the uncoded level can be made arbitrarily small when we choose $\Pi_{l=1}^L p_l$ sufficiently large. Therefore, one then has to show that the linear code $C^1\times\ldots\times C^L$ over $\mbb{F}_{p_1}\times\ldots\times \mbb{F}_{p_L}$ together with the mapping $\mc{M}$ is capacity-achieving for the $\Lambda/\Lambda'$ channel under multistage decoding.

Now, by the chain rule of mutual information \cite{cover91}, one has that
\begin{align}
    I(\msf{Y};\msf{X}) &= I(\msf{Y};\mc{M}(\msf{C}^1,\ldots,\msf{C}^L)) \nonumber \\
    &= I(\msf{Y};\msf{C}^1,\ldots,\msf{C}^L) = \sum_{l=1}^L I(\msf{Y};\msf{C}^l|\msf{C}^1,\ldots,\msf{C}^{l-1})).
\end{align}
Hence, the only task remained is showing that linear codes over $\mbb{F}_{p_l}$ can achieve the conditional mutual information $I(\msf{Y};\msf{C}^l|\msf{C}^1,\ldots,\msf{C}^{l-1})$. To this end, we follow the proof in \cite{forney2000} and show that the equivalent channel at each level is regular in the sense of Delsarte and Piret \cite{delsarte82}.

As restated in \cite{forney2000}, a channel with transition probabilities $\{f(y|b), b\in B, y\in Y\}$ is regular if the input alphabet can be identified with an Abelian group $B$ that acts on the output alphabet $Y$ by permutation. In other words, if a set of permutations $\{\tau_b, b\in B\}$ can be defined such that $\tau_b(\tau_{b'}(y))=\tau_{b\oplus b'}(y)$ for all $b, b'\in B$ and $y\in Y$ such that $f(y|b)$ depends only on $\tau_b(y)$. Note that since we are considering the $\Lambda/\Lambda'$ channel, the additive noise is actually the $\Lambda'$-aliased Gaussian noise given by
\begin{equation}
    f_{\Lambda'}(z) = \sum_{\lambda\in\Lambda} g_{\eta^2}(z+\lambda),~~z\in\mbb{R}^N,
\end{equation}
where $g_{\eta^2}(.)$ is the Gaussian density function with zero mean and variance $\eta^2$.

Now, suppose we are at the $l$th level's decoding. i.e., all the codewords in the previous levels have been successfully decoded. The receiver first subtracts out the contribution from the previous levels by $y-\mc{M}(c^1,\ldots,c^{l-1},0,\ldots,0)\mod\Lambda'$. We show that the equivalent channel seen at the $l$th level's decoding is regular. For $b\in\mbb{F}_{p_l}$ define
\begin{equation}
    \mathbf{b} \defeq \begin{bmatrix}
                        \mc{M}(0,\ldots,0,b,v_1^{l+1},\ldots,v_1^L) \\
                        \mc{M}(0,\ldots,0,b,v_2^{l+1},\ldots,v_2^L) \\
                        \vdots \\
                        \mc{M}(0,\ldots,0,b,v_S^{l+1},\ldots,v_S^L) \\
                      \end{bmatrix},
\end{equation}
where $(v_s^{l+1},\ldots,v_s^L)\in\mbb{F}_{p_{l+1}}\times,\ldots,\times\mbb{F}_{p_L}$ for $s\in\{1,\ldots,S\}$ and none of these vectors are exactly the same. Therefore, there are total $S=\Pi_{l'>l} p_{l'}$ possibilities. Also, note that the ordering of elements in $\mathbf{b}$ does not matter and can be arbitrarily placed. Thus, given the previous decoded codewords, $\mathbf{b}$ is fully determined by $b$. For $y\in\mbb{R}^N$, let us now define the following,
\begin{align}
    \tau_b(y)&\defeq y-\mathbf{b} \mod \Lambda' \nonumber \\
    &\defeq\begin{bmatrix}
                        y- \mc{M}(0,\ldots,0,b,v_1^{l+1},\ldots,v_1^L) \mod\Lambda'\\
                        y- \mc{M}(0,\ldots,0,b,v_2^{l+1},\ldots,v_2^L) \mod\Lambda' \\
                        \vdots \\
                        y- \mc{M}(0,\ldots,0,b,v_S^{l+1},\ldots,v_S^L) \mod\Lambda'\\
                      \end{bmatrix}.
\end{align}
One can verify that
\begin{align}
    \tau_b(\tau_{b'}(y)) &= \tau_{b'}(y) - \mathbf{b} \mod \Lambda' \nonumber \\
    &= \begin{bmatrix}
                        y- \mc{M}(0,\ldots,0,b',v_1^{l+1},\ldots,v_1^L) - \mc{M}(0,\ldots,0,b,v_1^{l+1},\ldots,v_1^L) \mod\Lambda'\\
                        y- \mc{M}(0,\ldots,0,b',v_2^{l+1},\ldots,v_2^L) - \mc{M}(0,\ldots,0,b,v_2^{l+1},\ldots,v_2^L) \mod\Lambda' \\
                        \vdots \\
                        y- \mc{M}(0,\ldots,0,b',v_S^{l+1},\ldots,v_S^L) - \mc{M}(0,\ldots,0,b,v_2^{l+1},\ldots,v_2^L) \mod\Lambda'\\
                      \end{bmatrix} \nonumber \\
    &\overset{(a)}{=} \begin{bmatrix}
                        y- \mc{M}(0,\ldots,0,b'\oplus b,2v_1^{l+1},\ldots,2v_1^L)\mod\Lambda'\\
                        y- \mc{M}(0,\ldots,0,b'\oplus b,2v_2^{l+1},\ldots,2v_2^L)\mod\Lambda' \\
                        \vdots \\
                        y- \mc{M}(0,\ldots,0,b'\oplus b,2v_S^{l+1},\ldots,2v_S^L)\mod\Lambda'\\
                      \end{bmatrix} \nonumber \\
                      &= \begin{bmatrix}
                        y- \mc{M}(0,\ldots,0,b'\oplus b,\tilde{v}_1^{l+1},\ldots,\tilde{v}_1^L)\mod\Lambda'\\
                        y- \mc{M}(0,\ldots,0,b'\oplus b,\tilde{v}_2^{l+1},\ldots,\tilde{v}_2^L)\mod\Lambda' \\
                        \vdots \\
                        y- \mc{M}(0,\ldots,0,b'\oplus b,\tilde{v}_S^{l+1},\ldots,\tilde{v}_S^L)\mod\Lambda'\\
                      \end{bmatrix}, \label{eqn:pi_regular}
\end{align}
where $(\tilde{v}_s^{l+1},\ldots,\tilde{v}_s^L)\in\mbb{F}_{p_{l+1}}\times,\ldots,\times\mbb{F}_{p_L}$ for $s\in\{1,\ldots,S\}$ and (a) follows from the fact that $\mc{M}$ is an isomorphism. Now, since $\mbb{Z}_p$ and $2\odot\mbb{Z}_p$ are isomorphic for all odd primes $p$, it is clear that none of $(\tilde{v}_s^{l+1},\ldots,\tilde{v}_s^L)$ for $s\in\{1,\ldots,S\}$ are the same so one can rearrange \eqref{eqn:pi_regular} to get $\tau_b(\tau_{b'}(y))=\tau_{b\oplus b'}(y)$.

Consider $b\in\mbb{F}_{p_l}$ is transmitted, the transition probability is given by
\begin{align}
    f(y|c^1,\ldots,c^{l-1},b) &\propto  \nonumber \\
    \sum_{(v^{l+1},\ldots,v^L)\in\mbb{F}_{p_{l+1}}\times,\ldots,\times\mbb{F}_{p_L}} &f_{\Lambda'}(y|c^1,\ldots,c^{l-1},b,v^{l+1},\ldots,v^L),
\end{align}
which only depends on $\tau_b(y)$. Hence the equivalent channel experienced by the $l$th level is regular and linear codes suffice to achieve the mutual information. Repeating this argument to each level shows that multilevel coding and multistage decoding suffice to achieve the capacity.

\subsection{Existence of MSE Quantization Good Lattices}\label{sec:MSE_good}
Recall that the normalized second moment $G(\Lambda)$ is invariant to scaling. Here, we choose to work with a scaled (by $\gamma q^{-1}$) lattice
\begin{equation}
    \Lambda = \gamma q^{-1}\mc{M}(C^1,\ldots,C^L) + \gamma\mbb{Z}^N,
\end{equation}
where $\gamma\defeq 2\sqrt{N\beta}$. For any $\mathbf{x}\in\mbb{R}^N$, define the MSE distortion
\begin{align}
    d(\mathbf{x},\Lambda) &= \frac{1}{N} \underset{\lambda\in\Lambda}{\min}\| \mathbf{x} - \lambda \|^2 \nonumber\\
    &= \frac{1}{N} \underset{\mathbf{a}\in\mbb{Z}^N, \mathbf{c}^l\in C^l, l\in\{1,\ldots,L\}}{\min}\| \mathbf{x} - \gamma q^{-1}\mc{M}(\mathbf{c}^1,\ldots,\mathbf{c}^L) -\gamma\mathbf{a} \|^2 \nonumber\\
    &= \frac{1}{N} \underset{\mathbf{c}^l\in C^l, l\in\{1,\ldots,L\}}{\min}\| \mathbf{x} - \gamma q^{-1}\mc{M}(\mathbf{c}^1,\ldots,\mathbf{c}^L) \mod \gamma\mbb{Z}^n \|^2.
\end{align}
For any $\mathbf{w}=[\mathbf{w}^1,\ldots,\mathbf{w}^L]$ where $\mathbf{w}_l\in\mbb{F}_{p_l}^{m^l}\setminus\{\mathbf{0}\}$, define $\mathbf{C}(\mathbf{w})\defeq [\mathbf{G}_1\odot\mathbf{w}^1,\ldots,\mathbf{G}_L\odot\mathbf{w}^L]$. Note that each $\mc{M}(\mathbf{C}(\mathbf{w}))$ is uniformly distributed over $\mbb{Z}^n/q\mbb{Z}^N$ as each $\mathbf{G}_l\odot\mathbf{w}^l$ is uniformly distributed over $\mbb{F}_{p_l}^N$ and $\mc{M}$ is a ring isomorphism. We can then follow \cite[(14)-(16)]{ordentlich_erez_simple} that for all $\mathbf{w}\in\times_{l=1}^L \mbb{F}_{p_l}^{m^l}\setminus\{\mathbf{0}\}$ and $\mathbf{x}\in\mbb{R}^N$,
\begin{align}
    \varepsilon &\defeq \Pp\left(\frac{1}{N} \| \mathbf{x} - \gamma q^{-1}\mc{M}(\mathbf{C}(\mathbf{w})) \mod \gamma\mbb{Z}^n \|^2 \leq \beta \right) \nonumber \\
    &\geq V_N 2^{-N}\left(1-\frac{\sqrt{N}}{q}\right)^N,
\end{align}
where $V_N$ is the volume of an $N$-dimensional ball with radius 1. Note that the reason that we exclude those all-zeros sub-messages is because those would make $\mathbf{C}(\mathbf{w})$ non-uniform and hence make the analysis more involoved. However, including those points will only help the quantization and hence, the above inequality is valid.

Now, let us choose $q=\xi N^{\frac{3}{2}}$ where $\xi$ is chosen to be the largest value in $[0.5,1)$ such that $q$ is a product of $L$ distinct primes. This in turns provides
\begin{equation}
    \varepsilon >\frac{1}{N^2}V_N2^{-N}.
\end{equation}
Following the similar probabilistic arguments in \cite[(18)-(19)]{ordentlich_erez_simple}, one has that
\begin{align}
    \Pp\left(d(\mathbf{x},\Lambda)>\beta\right) &< N^{\frac{7}{2}}\Pi_{l=1}^L\frac{1}{(p_l^{m^l}-1)}2^N V_N^{-1} \nonumber \\
    &< 2^L N^{\frac{7}{2}}\Pi_{l=1}^L p_l^{-m^l}2^N V_N^{-1}
\end{align}

Using the above inequality, one can bound the expectation of the distortion over the ensemble of lattices and $\mathbf{X}$ which may have arbitrary distribution as
\begin{equation}
    \mbb{E}_{\mathbf{X},\Lambda}(d(\mathbf{X},\Lambda))\leq \beta\left(1+2^L N^{\frac{9}{2}}2^{-N\left[\sum_{l=1}^L \frac{m^l}{N}\log(p_l)-\frac{1}{2}\log\left(\frac{4}{V_N^{2/N}}\right)\right]}\right),
\end{equation}
which tends to $\beta$ provided that
\begin{equation}\label{eqn:MSE_good_criterion}
    \sum_{l=1}^L \frac{m^l}{N}\log(p_l)=\frac{1}{2}\log\left(\frac{4}{V_N^{2/N}}\right) + \delta,
\end{equation}
for a $\delta>0$. As in \cite{ordentlich_erez_simple}, picking $\mathbf{X}$ to have uniform distribution over $\gamma[0,1)^N$ allows one to relate the MSE to the second moment of the lattice. Thus,
\begin{equation}\label{eqn:exp_2_moment}
    \underset{N\rightarrow\infty}{\lim}\mbb{E}_{\Lambda}(\sigma^2(\Lambda)) = \underset{N\rightarrow\infty}{\lim}\mbb{E}_{\mathbf{X},\Lambda}(d(\mathbf{X},\Lambda))\leq\beta,
\end{equation}
and hence $\lim_{N\rightarrow\infty}\mbb{E}_{\Lambda}(\sigma^2(\Lambda))\leq 1$.
Moreover, the volume of the normalized fundamental Voronoi region of $\Lambda$ is lower bounded by
\begin{align}
    \text{Vol}(\mc{V}_{\Lambda})^{\frac{2}{N}} &\overset{(a)}{\geq} (\gamma^N\Pi_{l=1}^L p_l^{-m^l})^{\frac{2}{N}} \nonumber \\
    &= 4N\beta\Pi_{l=1}^L p_l^{-2m^l/N} \nonumber \\
    &\overset{(b)}{=} 2^{-2\delta}N V_N^{\frac{2}{N}}\beta,
\end{align}
where (a) becomes an equality if and only if every $\mathbf{G}_l$ is full rank and (b) follows from the choice of \eqref{eqn:MSE_good_criterion}. With the expectation of the second moment and the volume of the normalized fundamental Voronoi region, one can then bound the expectation of the normalized second moment over the ensemble of lattices as
\begin{align}\label{eqn:lim_nsm}
    \lim_{N\rightarrow \infty}\mbb{E}_{\Lambda}(G(\Lambda)) &= \lim_{N\rightarrow \infty}\mbb{E}_{\Lambda}\left( \frac{\sigma^2(\Lambda)}{\text{Vol}(\mc{V}_{\Lambda})^{\frac{2}{N}}}\right) \nonumber \\
    &\leq 2^{2\delta}\lim_{N\rightarrow\infty}\frac{1}{NV_N^{\frac{2}{N}}} \nonumber \\
    &= 2^{2\delta}\frac{1}{2\pi\exp(1)}.
\end{align}
After this, by applying the Markov inequality, one obtains that with high probability, the lattice from the ensemble is good for MSE quantization asymptotically.

\subsection{Existence of Simultaneously-Good Lattices}
Although our proof for the achievable computation does not require simultaneously-good lattices, we show the existence of such lattices generated by the proposed product construction for the sake of completeness. To show the simultaneous goodness, one has to make sure that the conditions for which the above two properties hold would not conflict with each other. As mentioned before, the normalized second moment is invariant to scaling, so choosing $\gamma = 2\sqrt{N(\eta^2+\epsilon)}$ with $\epsilon>0$ and $\eta^2$ the variance of the Gaussian noise, will not change the result. Now, choosing $p_l$ and $m^l$ such that \eqref{eqn:MSE_good_criterion} ensures that with high probability the sequence of lattices is good for quantization asymptotically. i.e., $\text{Vol}(\mc{V}_{\Lambda})^{\frac{2}{N}} \rightarrow 2\pi\exp(1)(\eta^2+\epsilon)$ in the limit as $N\rightarrow \infty$. From the capacity separability in \cite[Theorem]{forney2000}, when choosing $q$ sufficiently large, one has that
\begin{align}
    \frac{1}{N}C(\Lambda/\Lambda',\eta^2) &\approx \frac{1}{N}C(\Lambda',\eta^2) \approx \frac{1}{2}\log\left(\frac{\text{Vol}(\mc{V}_{\Lambda'})^{\frac{2}{N}}}{2\pi\exp(1)\eta^2}\right) \nonumber \\
    &= \frac{1}{2}\log\left(\frac{\text{Vol}(\mc{V}_{\Lambda'})^{\frac{2}{N}}}{\text{Vol}(\mc{V}_{\Lambda})^{\frac{2}{N}}}\right) + \frac{1}{2}\log\left(\frac{\text{Vol}(\mc{V}_{\Lambda})^{\frac{2}{N}}}{2\pi\exp(1)\eta^2}\right) \nonumber \\
    & = \frac{1}{2}\sum_{l=1}^L\log\left( p_l^{2m^l/N}\right) + \frac{1}{2}\log\left(\frac{ 2\pi\exp(1)(\eta^2+\epsilon) }{2\pi\exp(1)\eta^2}\right)\nonumber \\
    &\approx \sum_{l=1}^L\log\left(p_l^{m^l/N}\right) + \frac{\epsilon}{\eta^2}.
\end{align}
Moreover, the approximation becomes exact in the limit as $N\rightarrow \infty$. This implies that when choosing parameters in \eqref{eqn:MSE_good_criterion} such that the lattices are good for quantization, the sum rates of the underlying linear codes would achieve the capacity $\frac{1}{N}C(\Lambda/\Lambda',\eta^2)$ asymptotically. However, this does not specify the rate for each level. Therefore, when $N$ large enough, one is free to pick the linear code for each level such that the rate is arbitrarily close to the capacity of that level. When doing so, as shown previously, the lattice would be Poltyrev-good under multistage decoding as well.

\section{Construction A Lattices over the Second Proposed Family of Constellations}\label{apx:simul_good}
In this appendix, we provide the definition of Construction A lattices with the second proposed family of constellations and provide a theorem about the existence of such lattices that are simultaneously good for MSE quantization and Poltyrev-good.

{\bf \underline{Construction A}} \cite{LeechSloane71} \cite{conway1999sphere} Let $\phi$ be an Eisenstein prime with $\phi$ being the product of a unit and a rational prime $q$ congruent to $2\mod 3$. Thus, one has $|\phi|^2 = q^2$. Let $k$, $N$ be integers such that $k\leq N$ and let $\mathbf{G}$ be the generator matrix of a $(N,k)$ linear code. Construction A consists of the following steps,
\begin{enumerate}
    \item Define the discrete codebook $\mc{C}=\{\mathbf{x}=\mathbf{G}\odot\mathbf{y}:\mathbf{y}\in\mbb{F}_{q^2}^k\}$ where all operations are over $\mbb{F}_{q^2}$.
    \item Generate the $N$-dimensional lattice $\Lambda_{\mc{C}}$ as $\Lambda_{\mc{C}}\triangleq \{\mathbf{\lambda}\in(\Zw)^N:\sigma(\mathbf{\lambda})\in\mc{C}\}$, where $\sigma$ is the homomorphism defined in \eqref{eqn:homo_prod_Eis1}.
    \item Scale $\Lambda_{\mc{C}}$ with $\phi^{-1}$ to obtain $\Lambda = \phi^{-1}\Lambda_{\mc{C}}$.
\end{enumerate}
Given $N,k,q$, we define an $(N,k,q)$ ensemble as the set of lattices obtained through Construction A described above where for each of these lattices, $\mathbf{G}_{ij}$ are i.i.d. with a uniform distribution over $\mbb{F}_{q^2}$.

\begin{theorem}
    A lattice drawn from the $(N,k,q)$ ensemble is simultaneously good for quantization and good for AWGN channel coding as $N\rightarrow \infty$ in probability as long as the parameters satisfy

    i) $k \leq \beta N$ for some $\beta<1$ but grows faster than $\log^2(N)$,

    ii) $k, q$ satisfy
    \begin{align}\label{eqn:q2k}
        q^{2k} &= \frac{(\sqrt{3}/2)^N}{\text{Vol}(\mc{B}(r_{\Lambda}^{\text{eff}}))} = \frac{(\sqrt{3}/2)^N\Gamma(N+1)}{\pi^N (r_{\Lambda}^{\text{eff}})^{2N}} \nonumber \\
        &\approx \sqrt{2N\pi}\left(\frac{\sqrt{3}}{2}\right)^N\left(\frac{2N}{2\exp(1)(r_{\Lambda}^{\text{eff}})^2}\right)^N,
    \end{align}
    and
    \begin{equation}\label{eqn:rmin}
        r_{min}<r_{\Lambda}^{\text{eff}}(N)<2r_{min},
    \end{equation}
    where $0<r_{min}<1/4$,

    iii) $\gamma\rightarrow 0$ and $\text{Vol}(\mc{V}_{\gamma\Lambda})$ remains constant.
\end{theorem}
\begin{proof}
    Following the steps in \cite{Engin12} with some modifications on the choice of parameters completes the proof.
\end{proof}

\section{Proof of Theorem~\ref{thm:sym_cap}}\label{apx:sym_rate}
\begin{IEEEproof}
    Without loss of generality, we can assume $\theta=0$. We prove this theorem by showing that the capacity region of the corresponding MAC channel has a symmetric shape. Therefore, the symmetric capacity always touches the boundary of the sum-rate limit.

    Recall that with the homomorphism given in \eqref{eqn:module_homo}, the transmitted signals are given by
    \begin{align}
        \mathbf{x}_k &= \gamma\left(\mathbf{c}_k^1 + \mathbf{c}_k^2 \omega\mod q\Zw\right) \nonumber \\
        &=\gamma\left((\mathbf{c}_k^1\mod q\mathbb{Z}) + \omega(\mathbf{c}_k^2\mod q\mathbb{Z})\right) \nonumber \\
        &=\gamma\left(\mathbf{\check{c}}_k^1 + \omega\mathbf{\check{c}}_k^2\right),
    \end{align}
    where $\mathbf{\check{c}}_k^l\triangleq \mathbf{c}_k^l\mod q\mbb{Z}$. Thus, the received signal can be rewritten as
    \begin{align}
        \mathbf{y} &= h_1\mathbf{x}_1 + h_2\mathbf{x}_2 + \mathbf{z} \nonumber \\
        &= h_1\gamma\left(\mathbf{\check{c}}_1^1 + \omega\mathbf{\check{c}}_1^2\right) + h_2 \gamma\left(\mathbf{\check{c}}_2^1 + \omega\mathbf{\check{c}}_2^2\right) +\mathbf{z}\nonumber \\
        &= \gamma\left(h_1\mathbf{\check{c}}_1^1 + h_2\mathbf{\check{c}}_2^1\right) + \gamma\left(h_1\mathbf{\check{c}}_1^2 + h_2\mathbf{\check{c}}_2^2\right)\omega +\mathbf{z}.
    \end{align}
    Notice that one also has
    \begin{equation}
        \omega\mathbf{\bar{y}} = \gamma\left(h_1\mathbf{\check{c}}_1^2 + h_2\mathbf{\check{c}}_2^2\right) +\gamma\left(h_1\mathbf{\check{c}}_1^1 + h_2\mathbf{\check{c}}_2^1\right)\omega + \omega\mathbf{\bar{z}},
    \end{equation}
    where $\omega\mathbf{\bar{z}}$ and $\mathbf{z}$ has the same distribution as $\mathbf{z}$ is circularly symmetric. Now, decoding second level first by regarding the first level as unknown would result in an information rate
    \begin{align}\label{eqn:level_1}
        I(\msf{Y};b_1\msf{C}_1^2\oplus b_2\msf{C}_2^2) &\overset{(a)}{=} I(\omega\msf{\bar{Y}};b_1\msf{C}_1^2\oplus b_2\msf{C}_2^2) \nonumber \\
        &\overset{(b)}{=}I(\msf{Y};b_1\msf{C}_1^1\oplus b_2\msf{C}_2^1)
    \end{align}
    where (a) follows from that the operations are bijective and (b) is due to the fact that $h_1\mathbf{\check{c}}_1^1 + h_2\mathbf{\check{c}}_2^1$ and $h_1\mathbf{\check{c}}_1^2 + h_2\mathbf{\check{c}}_2^2$ are statistically the same.

    Similarly, decoding the first level by assuming the decoded second level is correct results in an information rate
    \begin{equation}\label{eqn:level_2}
        I(\msf{Y};b_1\msf{C}_1^1\oplus b_2\msf{C}_2^1|b_1\msf{C}_1^2\oplus b_2\msf{C}_2^2) = I(\msf{Y};b_1\msf{C}_1^2\oplus b_2\msf{C}_2^2|b_1\msf{C}_1^1\oplus b_2\msf{C}_2^1).
    \end{equation}
    Now, combining \eqref{eqn:level_1} and \eqref{eqn:level_2}, one can show that the capacity region of this MAC channel is symmetric. This completes the proof.
\end{IEEEproof}

\bibliographystyle{ieeetr}
\bibliography{journal_abbr,prod}

\end{document}